\documentclass[journal]{IEEEtran}

\ifCLASSINFOpdf

\else

\fi

\usepackage{graphicx}	
\usepackage{amsmath}
\usepackage{mathtools,amssymb}
\usepackage{amsthm}
\usepackage{bbm}
\usepackage{mathtools}
\usepackage{tabularx,algorithm}	
\usepackage{color}
\usepackage{algpseudocode}
\usepackage{bm}
\usepackage{tikz}
\usepackage{epstopdf}
\usepackage{cite}
\usepackage{amsfonts}
\usepackage{dsfont}
\usepackage{amssymb}
\usepackage{algorithm}
\usepackage{epsfig}
\usepackage[labelformat=simple]{subcaption}

\newcommand{\overbar}[1]{\mkern 2.5mu\overline{\mkern-2.5mu#1\mkern-2.5mu}\mkern 2.5mu}

\newtheorem{theorem}{Theorem}
\newtheorem{definition}{Definition}

\newtheorem{lemma}{Lemma}
\newtheorem{remark}{Remark}

\newtheorem{example}{Example}
\newtheorem{assumption}{Assumption}

\usetikzlibrary{positioning}
\usetikzlibrary{decorations.pathreplacing}
\usetikzlibrary{shapes,arrows}

\begin{document}

\title{Global and Asymptotically Efficient Localization from Range Measurements}
\author{Guangyang Zeng, Biqiang Mu, Jiming Chen, Zhiguo Shi, and Junfeng Wu
	\thanks{Guangyang Zeng and Jiming Chen are with the College of Control Science and Engineering and the State Key Laboratory of Industrial Control Technology, Zhejiang University, Hangzhou 310027, P. R. China.
		{\tt\small \{gyzeng,cjm\}@zju.edu.cn}.}
	\thanks{Biqiang Mu is with Key Laboratory of Systems and Control, Institute of Systems Science, Academy of Mathematics and Systems Science, Chinese Academy of Sciences, Beijing 100190, China.
		{\tt\small bqmu@amss.ac.cn}.}
	\thanks{Zhiguo Shi is with the College of Information Science and Electronic Engineering, Zhejiang University, Hangzhou 310027, P. R. China.
		{\tt\small shizg@zju.edu.cn}.}
	\thanks{Junfeng Wu is with the School of Data Science, Chinese University of Hong Kong, Shenzhen, Shenzhen, P. R. China, and the State Key Laboratory of Industrial Control Technology, Zhejiang University, Hangzhou, P. R. China.
		{\tt\small junfengwu@cuhk.edu.cn}.}
}

\maketitle

\begin{abstract}
We consider the range-based localization problem, which involves estimating an object's position by using $m$ sensors, hoping that as the number $m$ of sensors increases, the estimate converges to the true position with the minimum variance. 
We show that under some conditions on the sensor deployment and measurement noises, the LS estimator is strongly consistent and asymptotically normal. However, the LS problem is nonsmooth and nonconvex, and therefore hard to solve. We then devise realizable estimators that possess the same asymptotic properties as the LS one. 
These estimators are based on a two-step estimation architecture, which says that any $\sqrt{m}$-consistent estimate followed by a one-step Gauss-Newton iteration can yield a solution that possesses the same asymptotic property as the LS one. The keypoint of the two-step scheme is to construct a $\sqrt{m}$-consistent estimate in the first step. In terms of whether the variance of measurement noises is known or not, we propose the Bias-Eli estimator (which involves solving a generalized trust region subproblem) and the Noise-Est estimator (which is obtained by solving a convex problem), respectively. Both of them are proved to be $\sqrt{m}$-consistent. Moreover, we show that by discarding the constraints in the above two optimization problems, the resulting closed-form estimators (called Bias-Eli-Lin and Noise-Est-Lin) are also $\sqrt{m}$-consistent. 
Plenty of simulations verify the correctness of our theoretical claims, showing that the proposed two-step estimators can asymptotically achieve the Cramer-Rao lower bound. 
\end{abstract}


%
\IEEEpeerreviewmaketitle

\section{Introduction}\label{section:intro}
\subsection{Background}
\IEEEPARstart{L}{ocalization} is a fundamental module in the extensive location-based services, including navigation systems, indoor robots, advertising, to name a few. It is referred to acquiring the position of an object with respect to a certain coordinate system based on the measurements of a set of sensors. The commonly used types of measurements include received signal strength (RSS)~\cite{liu2015rss}, angle of arrival (AOA)~\cite{sun2020eigenspace}, time of arrival (TOA)~\cite{shen2012accurate} and time difference of arrival (TDOA)~\cite{sun2018solution}. 
In virtue of the development of high-precision time measurement methods, e.g., the ultra-wideband (UWB) technology~\cite{wymeersch2012machine}, the TOA-based localization can achieve a quite accurate estimation of the object position, and it has been applied widely in the Internet of Things. In this paper, we will investigate the object localization from TOA measurements. 

The TOA methods utilize the absolute time instants when a radio signal emanating from the object reaches several sensors. With the prior knowledge of the propagation velocity of the signal, the range measurements from the object to sensors are obtained. Each TOA measurement will narrow the object position to a circle (in the 2-D case) or a sphere (in the 3-D case) centered at the sensor. Due to the presence of measurement noises, these circles (spheres) do not necessarily intersect at one point. Therefore, advanced estimation methods are required to estimate the object's position.

\subsection{Related Works}
Maximum likelihood (ML) is the most natural criterion in parameter inference for the property of statistical interpretation and asymptotic efficiency. Hence, there exists extensive literature studying the TOA-based localization under the ML formulation~\cite{chan1994simple,chan2006exact,soares2015simple,biswas2006semidefinite,biswas2006semidefinite2}. However, due to the nonconvex property of the ML problem, its global minimum is hard to find. Iterative schemes are adopted in~\cite{chan1994simple,chan2006exact} to seek approximate solutions. These methods are usually sensitive to the initial estimate and generally converge to local minima. There are also some semidefinite relaxation (SDR) approaches which relax the nonconvex problem into a convex one~\cite{soares2015simple,biswas2006semidefinite,biswas2006semidefinite2}.
Although the global minimizers of these relaxed problems can be found, they are not optimal in the ML sense. 

Another prevalent criterion is least squares (LS), based on which the sum of squared errors is minimized~\cite{caffery1998subscriber,wymeersch2009cooperative,gustafsson2005mobile,kim2006interior,caffery2000new,cheung2004accurate,beck2008exact}. If the noises follow i.i.d. Gaussian distribution, the LS estimation resembles the ML one. It is noteworthy that the TOA LS problem is also nonconvex, and gradient-based methods~\cite{caffery1998subscriber,wymeersch2009cooperative,gustafsson2005mobile} do not necessarily obtain the global minimum. Linear approximation~\cite{kim2006interior,caffery2000new} and SDR~\cite{cheung2004accurate,beck2008exact,zhou2018fast} are two widely adopted attempts to obtain approximate solutions. 
An important approach that emerged from the LS formulation is to apply the least squares methodology to the squared range measurements, which is called the squared least squares (S-LS) formulation~\cite{beck2008exact,cheung2004least}. Although the S-LS problem is nonconvex, it owns an unrivaled property that its global minimum can be obtained via solving a generalized trust region subproblem (GTRS)~\cite{beck2008exact}. The Lagrangian dual approach proposed in~\cite{qi2013lagrangian} is equivalent to the GTRS method. There is some literature discussing the relationship between the LS and the S-LS problems~\cite{larsson2009accuracy,beck2008iterative,beck2012solution}. Larsson and Danev~\cite{larsson2009accuracy} compared the asymptotic accuracy between the LS and the S-LS estimators. They identified geometries where the performances of the two methods are identical but also geometries when the difference in performance is unbounded. Beck \textit{et al.}~\cite{beck2008iterative,beck2012solution} empirically showed that the S-LS solution provides a ``good'' initial estimate for some iterative algorithms to globally solve the LS problem. 


\subsection{Contributions and Organization}
Based on our literature survey, we find that the realization of the LS estimator by solving the LS optimization problem~\eqref{LS} is nontrivial~\cite{caffery1998subscriber,wymeersch2009cooperative,gustafsson2005mobile,chuang2008effective,wu2019adaptive}. All of these works cannot guarantee to obtain the global minimum. In addition, most of the literature which investigates the modified or relaxed counterparts of the LS problem does not analyze the relationship between the proposed estimator and the LS solution~\cite{kim2006interior,caffery2000new,cheung2004accurate,beck2008exact,qi2013lagrangian}. Moreover, the consistency and asymptotic normality of the LS estimation are rarely discussed in the literature. 

In this paper, we show that under some mild and readily-checked conditions (Assumptions~\ref{assumption:Gaussian_noise}-\ref{nonsingular_assumption}), the LS estimator is strongly consistent and asymptotically normal. Further, we claim that although the LS estimator is hard to realize, it is viable to devise estimators that possess the same asymptotic property as the LS one. 
Specifically, we propose such estimators in virtue of a two-step estimation architecture~\cite{lehmann2006theory,mu2017globally}. The two-step architecture has the unrivaled property that given any $\sqrt{m}$-consistent estimate in the first step, just a one-step Gauss-Newton (GN) iteration in the second step, which is computationally efficient, can yield an estimate with the minimum variance (asymptotically). The keypoint of the two-step architecture is to devise a $\sqrt{m}$-consistent estimator in the first step. 
In this paper, we show the possibility of constructing $\sqrt{m}$-consistent estimators by solving modified S-LS problems. In terms of whether the statistical knowledge (variance) of measurement noises is known or not, we propose different $\sqrt{m}$-consistent estimators. When the variance is available, we subtract it from the error terms of the S-LS problem, yielding the Bias-Eli problem~\eqref{GTRS_problem}. The Bias-Eli optimization problem is a GTRS problem, and we propose a complete algorithm to seek its global minimizer. In addition, we discard the constraint in~\eqref{GTRS_problem}, yielding an ordinary least squares problem~\eqref{linear_Bias_Eli} called the Bias-Eli-Lin problem. Both the Bias-Eli and Bias-Eli-Lin estimators are proved to be $\sqrt{m}$-consistent. When the statistical knowledge of noises is unknown, we turn to simultaneously estimate the object's position and noises' variance. We first propose a $\sqrt{m}$-consistent estimator by solving the Noise-Est problem~\eqref{Noise_Estimate_problem2}, which is a convex one. We also show that the solution to the ordinary least squares problem Noise-Est-Lin~\eqref{Noise_estimate_linear} which discarding the constraint in~\eqref{Noise_Estimate_problem2} is $\sqrt{m}$-consistent. 
Note that some of the existing localization systems have a high speed of measurements, e.g., sampling rate of practical UWB systems can reach $2.3$ kHz~\cite{grossiwindhager2019snaploc}. Therefore, the proposed asymptotically optimal estimators may play a valuable role in these systems for line-of-sight (LOS) scenarios, especially when the target is static, and a large sample of measurements can be utilized.
To summarize, the main contributions of this paper are listed as follows:
\begin{enumerate}
	\item [$(i).$] Under some assumptions (Assumptions~\ref{assumption:Gaussian_noise}-\ref{nonsingular_assumption}), we prove that the LS estimator which involves solving the nonsmooth and nonconvex LS optimization problem~\eqref{LS} is strongly consistent and asymptotically normal. These assumptions are stated tightly related to practical settings rather than mere abstract mathematical ones, which can be readily checked. Specifically, they are associated with range measurement noises and the deployment of sensors. Some examples of sensor arrangement are supplied to make them more comprehensible. 
	
	\item [$(ii).$] Noticing the LS problem~\eqref{LS} is hard to solve, we devise realizable estimators that possess the same asymptotic properties as the LS one. The estimators are based on a two-step estimation architecture where the keypoint is to construct a $\sqrt{m}$-consistent estimator in the first step. When the variance of measurement noises is known, we devise the Bias-Eli estimator which involves solving a GTRS problem~\eqref{GTRS_problem}. Otherwise, we construct the Noise-Est estimator, which is obtained by solving a convex problem~\eqref{Noise_Estimate_problem2}, to simultaneously estimate the object's position and noises' variance. Both estimators are proved to be $\sqrt{m}$-consistent.  
	
	\item [$(iii).$] By discarding the constraints in the GTRS problem~\eqref{GTRS_problem} and the convex problem~\eqref{Noise_Estimate_problem2}, we obtain two ordinary least squares problems~\eqref{linear_Bias_Eli} and~\eqref{Noise_estimate_linear}, which yield the Bias-Eli-Lin estimator and the Noise-Est-Lin estimator, respectively. In virtue of their closed-form expressions, we prove that both estimators are also $\sqrt{m}$-consistent. In addition, we study their biases and MSEs in the finite sample case, showing that they are unbiased and have the same MSEs.
\end{enumerate}

We remark that our work brings deeper insights into the real applications of TOA-based localization. The consistency and asymptotic normality of the LS estimator, which highly depend on the deployment of sensors, have been overlooked and not yet fully discussed in most of the literature. In addition, we show that some ordinary least squares solutions obtained by discarding constraints of original optimization problems own some elegant properties, e.g., unbiasedness and consistency, which is usually neglected in the existing works. Based on the theoretical developments in this paper, one can appropriately deploy the sensors and adopt computationally efficient algorithms to obtain an estimate which converges to the true object's coordinates with the minimum variance.

The rest of the paper is organized as follows. 
In Section~\ref{model_and_LS_formulation}, we introduce the range measurement model and the formulations of the LS and S-LS problems. 
In Section~\ref{property_of_LS_estimator}, we give some assumptions and show the consistency and asymptotic normality of the LS estimator. 
In Section~\ref{two_step_estimators}, we introduce a two-step estimation scheme to achieve the same asymptotic properties as the LS estimator and focus on devising $\sqrt{m}$-consistent estimators in the first step.
In Section~\ref{finite_sample_section}, we derive the biases and MSEs of the proposed first-step estimators which have closed-form expressions in the finite sample case. 
Simulation results are presented in Section~\ref{simulations}, followed by conclusions in Section~\ref{conclusion}.

{\bf Notations:} For a vector $x \in \mathbb R^n$, $[x]_i,i=1,\ldots,n$, presents the $i$-th element of $x$. For two vectors $x,y \in \mathbb R^n$, $x \preceq y$ denotes the pointwise inequality. Let $p=(p_i)_{i\in\mathbb N}$ and $q=(q_i)_{i\in\mathbb N}$ be two sequences of real numbers. When ${i\in\mathbb N}$ is clear from the context, we will omit the subscript and write $(p_i)$ as  a shorthand of  $(p_i)_{i\in\mathbb N}$. If $t^{-1} \sum_{i=1}^{t} p_i q_i$ converges to a real number its limit $\left\langle p,q \right\rangle_t $ will be called the tail product of $p$ and $q$.
We call $\|p\|_t=\sqrt{\langle p,p
	\rangle_t}$,
if it exists, the tail norm of $p$. For a sequence $p=(p_i)$ and a scalar $a$, $p-a$ produces a sequence, of which the $i$-th element is $p_i-a$.
For a cumulative distribution function $F_{\mu}$, $\mu$ is the measure induced from $F_{\mu}$.

\section{Localization Problem with Range Measurements} \label{model_and_LS_formulation}

\subsection{Range Measurement and Least Squares (LS) Problem}
Let $x^o \in \mathbb R^n$ be the coordinates of the object 
(target to be localized) and $a_i \in \mathbb R^n, i=1,\ldots ,m$ be the coordinates of the sensors. The measured distance from the object to the $i$-th sensor is denoted as $d_i$,
which consists of the true range plus a measurement noise as shown in the following equation:
\begin{equation}\label{range_measurement_model}
d_{i}=\|a_i-x^o\|+r_{i}.
\end{equation}The geometry of the range measurement model is illustrated in Fig.~\ref{range_measurements}.
In the following, we make an assumption on the statistics of the
measurement noises:
\begin{assumption}\label{assumption:Gaussian_noise}
	The $r_{i}$'s are i.i.d. Gaussian noises with 0 mean and finite variance $\sigma^2$.
\end{assumption}

It is desired to find an estimate of the true position $x^o$ such that the estimate fits the collected measurements best in the least squares sense.  A least squares (LS) estimate of $x^o$ based on the measurements $d_1,\ldots, d_m$ is an optimal solution to the following problem:
\begin{equation}\label{LS}
\hbox{(\textbf{LS}):}~~~~\mathop{\rm minimize}_{x\in\mathbb R^n} \sum_{i=1}^{m}  \left(\|a_i-x\|-d_{i}\right)^2.
\end{equation}
The LS problem~\eqref{LS} is nonsmooth and nonconvex, and finding its global minimum is not a simple task. There are basically two directions to do so. The first one attempts to do the global exploration through grid search and random search techniques, e.g., the particle swarm optimization~\cite{chuang2008effective}. The other one corresponds to the local search which needs a good initial estimate, e.g., the Gauss-Newton method~\cite{wu2019adaptive}. There are also some convex relaxation methods that seek solutions to the resulting relaxed problems~\cite{cheung2004accurate,beck2008exact}. However, we remark that all of the aforementioned methods cannot guarantee the global minimum of~\eqref{LS}, and therefore fail in finding an LS estimate of $x^o$.
\begin{figure}[!htb]
	\centering
	\includegraphics[width=0.26\textwidth]{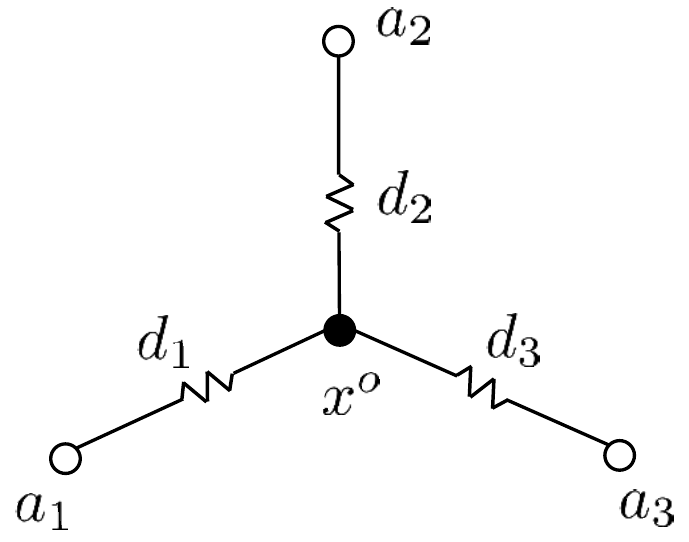}
	\caption{Illustration of $2$D range measurements of a sensor network consisting of three sensors. The solid dot ``$\bullet$'' represents the object, whose coordinate is $x^o$, and the hollow dots ``$\circ$'' represent range sensors, whose coordinate are $a_i$'s. The range measurement (jagged line means that the measurements are noisy) between the $i$-th sensor and the object is denoted as $d_i$.}
	\label{range_measurements}
\end{figure}

\subsection{Squared Range Model and Some Existing Localization Methods}

As mentioned in the previous subsection, due to the nonsmooth and nonconvex properties, the global minimum of~\eqref{LS} is hard to seek. An alternative way is to study a model that is a variation of~\eqref{range_measurement_model}. To do so, 
by squaring~\eqref{range_measurement_model}, it gives the squared range model as follows:
\begin{equation}\label{squared_range_measurements}
d_{i}^2=\|a_i-x^o\|^2+e_{i}, 
\end{equation}
where $e_{i}=2\|a_i-x^o\| r_{i}+r_{i}^2$ are the error terms. Then the squared least squares (S-LS) problem is constructed as:
\begin{equation}\label{S-LS}
\hbox{(\textbf{S-LS}):}~~~~\mathop{\rm minimize}_{x\in\mathbb R^n} \sum_{i=1}^{m} \left(\|a_i-x\|^2-d_{i}^2\right)^2.
\end{equation}

The S-LS problem is still nonconvex. There have been extensive studies on the squared range model. Representative methods include linear correction~\cite{huang2001real}, quadratic-term elimination~\cite{li2004least}, Lagrange multiplier~\cite{cheung2004least,qi2013lagrangian}, etc, resulting in a variety of basic localization algorithms.  Among these works, Beck \textit{et al.}~\cite{beck2008exact} equivalently transformed the S-LS problem into a generalized trust region subproblem (GTRS) which consists of a quadratic objective function and a quadratic equality constraint. In virtue of the necessary and sufficient condition of GTRS solutions derived in~\cite{more1993generalizations}, a global minimizer seeking algorithm was proposed in~\cite{beck2008exact}. Some other works~\cite{cheung2004least,li2004least} discarded the quadratic constraint, which gives rise to the unconstrained LS (ULS) problem with a closed-form solution.  

For the LS problem~\eqref{LS}, it is not clear from our literature survey that under what conditions the LS solution is consistent and asymptotically normal. Besides, we know from~\eqref{squared_range_measurements} that the mean of $e_i$ is $\sigma^2$ (not $0$), which makes the S-LS estimator biased, even in the asymptotic case where $m$ goes to infinity (this will be shown in Fig.~\ref{asymptotic:Bias_S_LS} in the simulations). Therefore the S-LS estimator is not consistent, and it is not feasible to construct an estimator that ``equals'' the LS one asymptotically based on the S-LS solution. All these lead to the interests of this paper.

\subsection{Problems of Interest}
\textit{First, we are interested in offering some readily-checked conditions to ensure the consistency and asymptotic normality of the LS estimator.} Moreover, observing that the global minimum of the LS problem~\eqref{LS} is hard to seek, and the existing localization algorithms (to the best of our knowledge) cannot attain the same (asymptotic) statistical performance as the LS estimator, \textit{another issue of this paper is to devise realizable estimators that possess the same asymptotic statistical properties as the LS estimator, i.e., an estimator achieves the CRLB in the large sample case.}

\section{Asymptotic Efficiency of the LS Estimator} \label{property_of_LS_estimator}

In this section, we will show that under some readily-checked assumptions, the LS estimator which solves~\eqref{LS} is strongly consistent and asymptotically normal. 
For model~\eqref{range_measurement_model}, we make the following assumption on the positions of the target and sensors.

\begin{assumption} \label{compact_set_interior}
	The true object's position $x^o$ is an interior point of a compact set $\mathcal X$. The positions of sensors $a_i$'s belong to a bounded set $\mathcal A$, and for each $i$, $a_i \notin \mathcal X$.
\end{assumption}

The condition that $x^o$ is an interior point of a compact set $\mathcal X$ is mild in practice. For example, it holds when the target is within the convex hull of the set of sensors, in which case the convex hull is identified as the compact set $\mathcal X$. Next, we give a definition of the sample distribution function and make an assumption on the sample distribution function of sensors' coordinates.

\begin{definition}
	The sample distribution function $F_m$ of a sequence $(z_1,z_2,\ldots)$ in $\mathbb R^n$ is defined as $F_m(z)=\#/m$ where $\#$ is the number of vectors in the subsequence $(z_1,\ldots,z_m)$ that satisfy $z_i \preceq z$.
\end{definition}
\begin{assumption} \label{convergence_of_sample_distribution}
	The sample distribution function $F_m$ of the sequence $(a_1,a_2,\ldots)$ converges to a distribution function $F_{\mu}$, i.e., $\lim\limits_{m \rightarrow \infty} F_m(a)=F_{\mu}(a)$, for all $a \in \mathbb R^n$ at which $F_{\mu}$ is continuous. We denote the probability measure generated by $F_{\mu}$ as $\mu$.
\end{assumption}

In what follows, we give two examples of sensor deployment that satisfy Assumption~\ref{convergence_of_sample_distribution}. 
\begin{example} \label{example_random_vector}
	When $a_i, i=1,\ldots,m$ are independent realizations of some random vectors with identical distribution function $F_{\mu}$, we have $\lim\limits_{m \rightarrow \infty} F_m(a)=F_{\mu}(a)$ for all $a \in \mathbb R^n$. 
\end{example}

\begin{example} \label{example_fix_sensors}
	Suppose the number of sensors $M$ is fixed, and each sensor makes $T$ i.i.d. measurements. In this manner, totally $MT$ TOA measurements can be used. This setting is realistic when the object is static or the sampling of the TOA measurements is sufficiently fast compared to the object motion. In this setup, as $T$ goes to infinity, $F_m$ converges to $F_{\mu}$, where $\mu(a_i)=1/M$ for any $i$.
\end{example}

Let $f_i(x):=\|a_i-x\|$ and $f(x):=\left(f_i(x) \right) $.
Based on Assumptions~\ref{compact_set_interior} and \ref{convergence_of_sample_distribution}, the tail norm $\left\| f(x)-f(x^o) \right\|_t$ is well defined. See detailed arguments in Appendix~\ref{existence_of_tail_product_norm}. 
In the following assumption, we suppose the configuration of the sensor network is elaborated.
\begin{assumption} \label{unique_solution_assumption}
	The function $\left\| f(x)-f(x^o) \right\|^2_t$ has a unique minimum at $x=x^o$. Or equivalently, we say that the target is asymptotically uniquely localizable.
\end{assumption}
In finite sample cases where $m$ is fixed, the concept of unique localizability means that $\sum_{i=1}^{m} \left(f_i(x)-f_i(x^o) \right)^2 $ has a unique minimum at $x=x^o$. A sufficient condition for unique localizability is that the sensors are in general positions, i.e., they do not locate on a line for $2$D localization, or do not scatter in a plane for $3$D localization~\cite{aspnes2006theory}. Note that the asymptotically unique localizability in Assumption~\ref{unique_solution_assumption} means that $\lim\limits_{m \rightarrow \infty} m^{-1}\sum_{i=1}^{m} \left(f_i(x)-f_i(x^o) \right)^2 $ has a unique minimum at $x=x^o$. It is a notion related to the whole sequence $(f_1,\ldots,f_m,\ldots)$, but independent of any finite number of $f_i$.  Therefore, any unique localizability property for any finite $m$ does not necessarily imply asymptotically unique localizability. Here we give an example.
\begin{example} \label{example_of_asymptotic_localizability}
	As shown in Fig.~\ref{figure_of_asymptotic_localizability}, there are three sensors denoted as $a_1$, $a_2$, and $a_3$. Suppose $a_1$ and $a_2$ each makes $T$ measurements, and $a_3$ makes only one measurement. Since $a_1$, $a_2$, and $a_3$ are in general positions, the target $x^o$ is uniquely localizable. However, $\mu(a_3)$ decays to $0$ as $T \rightarrow \infty$, making that $\left\| f(x)-f(x^o) \right\|^2_t$ has two minima: one is the true position $x^o$, the other is the red point. Hence, the target $x^o$ is not asymptotically uniquely localizable.
	\begin{figure}[!htb]
		\centering
		\includegraphics[width=0.3\textwidth]{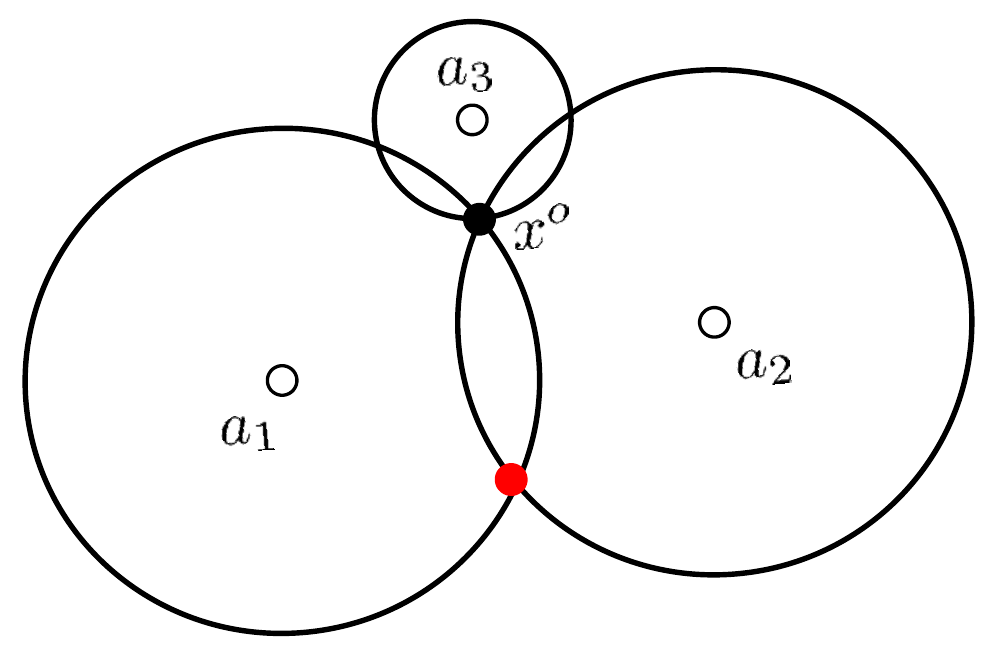}
		\caption{An example of unique localizability not implying asymptotically unique localizability.}
		\label{figure_of_asymptotic_localizability}
	\end{figure}
\end{example}

\begin{assumption} \label{nonsingular_assumption}
	There does not exist a line $\mathcal S$ for 2D localization (a plane $\mathcal S$ for 3D localization) such that $x^o \in \mathcal S$ and $\mu(\mathcal S)=1$.
\end{assumption}

Note that when the unknown $x^o$ is fixed, Assumptions~\ref{unique_solution_assumption} and~\ref{nonsingular_assumption} are essentially the assumptions on sensor deployment. In the following, we make a unified assumption on sensor deployment to guarantee Assumptions~\ref{unique_solution_assumption} and~\ref{nonsingular_assumption}. We remark that this unified assumption generally holds in real localization scenarios and can be readily checked. 

	\begin{assumption} \label{deployment_of_sensor}
		There does not exist a line $\mathcal S$ for 2D localization (a plane $\mathcal S$ for 3D localization) such that $\mu(\mathcal S)=1$.
	\end{assumption}

Given Assumption~\ref{deployment_of_sensor}, Assumption~\ref{nonsingular_assumption} holds straightforwardly. Next, we give a lemma to show that Assumption~\ref{deployment_of_sensor} implies Assumption~\ref{unique_solution_assumption}. 
\begin{lemma} \label{Assumption_6_to_4}
	Given Assumption~\ref{deployment_of_sensor}, the target is asymptotically uniquely localizable. 
\end{lemma}
\begin{proof}
	By definition, we have 
	\begin{equation*}
	\left\| f(x)-f(x^o) \right\|^2_t=\mathbb E_{a \sim \mu}\left[\left(\|a-x\|-\|a-x^o\| \right) ^2 \right], 
	\end{equation*}
	where $\mathbb E_{a \sim \mu}$ is taken over $a$ with respect to $\mu$, and $\left\| f(x^o)-f(x^o) \right\|^2_t=0$. For any $x \in \mathcal X$, define $\mathcal A_x=\{a\in \mathcal A \mid \|a-x\|=\|a-x^o\|\}$. Suppose there is an $x' \neq x^o$ such that $\mathbb E_{a \sim \mu}\left[\left(\|a-x'\|-\|a-x^o\| \right) ^2 \right]=0$. Then, $\mu(\mathcal A_{x'})=1$. Note that $\mathcal A_{x'}$ is the vertical bisector of the segment connecting $x^o$ and $x'$. Hence $\mathcal A_{x'}$ is a line in 2D case and a plane in 3D case. This contradicts Assumption~\ref{deployment_of_sensor}. Hence, $\left\| f(x)-f(x^o) \right\|^2_t$ has a unique minimum at $x=x^o$.
\end{proof}

Let $f_j'(x)=\left( f_{ji}'(x)\right) $ and $f_{jk}''(x) =\left( f_{jki}''(x)\right) $ be two sequences with respect to $i$, where 
\begin{equation*}
f_{ji}'(x)=\frac{\partial f_i(x)}{\partial [x]_j}~~\text{and}~~f_{jki}''=\frac{\partial^2 f_i(x)}{\partial [x]_j \partial [x]_k}.
\end{equation*}
Denote the optimal solution to the LS problem~\eqref{LS} as $\hat x^{\rm LS}_{m}$. Under Assumptions~\ref{assumption:Gaussian_noise}-\ref{nonsingular_assumption} or Assumptions~\ref{assumption:Gaussian_noise}-\ref{convergence_of_sample_distribution},\ref{deployment_of_sensor}, the LS estimator $\hat x^{\rm LS}_{m}$ enjoys the following consistency and asymptotic normality. The proof of Theorem~\ref{theorem_LS} is presented in Appendix~\ref{proof_theorem_ls}.

\begin{theorem}[Consistency and asymptotic normality] \label{theorem_LS}
	Under Assumptions~\ref{assumption:Gaussian_noise}-\ref{nonsingular_assumption} or Assumptions~\ref{assumption:Gaussian_noise}-\ref{convergence_of_sample_distribution},\ref{deployment_of_sensor}, we have $\hat x^{\rm LS}_{m} \rightarrow x^o$ with probability one as $m \rightarrow \infty$. Moreover,
	\begin{equation}\label{asymptotic_normality_LS}
	\sqrt{m}(\hat x^{\rm LS}_{m}-x^o) \rightarrow \mathcal N(0,\sigma^2 M^{-1}(x^o))~~\text{as}~m \rightarrow \infty,
	\end{equation}
	where $M(x) \in \mathbb R^{n \times n}$ is a matrix of which $M_{jk}(x)=\left\langle f_j'(x),f_k'(x)\right\rangle_t $.
\end{theorem}

		The Fisher information matrix $F$ of model~\eqref{LS} is given in~\eqref{Fisher_information_matrix}, and the CRLB is ${\rm tr}(F^{-1})$. From~\eqref{asymptotic_normality_LS} we have that the covariance of $\hat x^{\rm LS}_{m}$ converges to $\frac{\sigma^2}{m} M^{-1}(x^o)$. Further combining the definition of $M(x)$, it holds that $\lim\limits_{m \rightarrow \infty} \frac{\sigma^2}{m} M^{-1}(x^o)=F^{-1}$, which implies that the LS estimator $\hat x^{\rm LS}_{m}$ is asymptotically efficient.

\begin{remark} \label{compared_with_asymptotic_literature}
	Larsson and Danev~\cite{larsson2009accuracy} derived the asymptotic performance of the LS estimator using the setup in Example~\ref{example_fix_sensors}. The analysis therein uses the consistency of the LS estimator and a first-order Taylor expansion to calculate the mean squared error. However, the conditions that guarantee consistency are not specifically discussed. We remark that Assumption~\ref{convergence_of_sample_distribution} in this paper, which includes Example~\ref{example_fix_sensors} as a special case, is much less conservative, therefore covering a wider range of localization scenarios. 
\end{remark}

\section{Global and Asymptotically Efficient Localization} \label{two_step_estimators}
\subsection{Two-step Method in Nonlinear Least-Squares Estimation}

The LS estimator involves seeking the global minimum of the nonconvex LS problem~\eqref{LS}. When using the gradient-based optimization methods, it requires that the starting point is within the attraction neighborhood of the optimal solution. Thus, the gradient-based optimization algorithm is generally applied to improve the precision when a good initial estimate, which is close to the global minimizer, has been obtained since the objective function in~\eqref{LS} can be locally approximated as a convex function in a small neighborhood of the global minimizer. 
Therefore, the search of the LS estimator $\hat x^{\rm LS}_{m}$ is often done in two steps~\cite{gourieroux1995statistics}: 

\vspace{2mm}
\noindent\fbox{%
	\parbox{0.47\textwidth}{%
		{\bf Step 1}. Determine a consistent but not necessarily precise estimate.
		
		{\bf Step 2}. Use this preliminary estimate as an initial value for some algorithms that determine the LS estimator.
	}%
}
\vspace{1mm}

In Step 2, the Gauss-Newton (GN) algorithm or other local methods are commonly used for improving the accuracy of the consistent estimate obtained in Step 1. The GN algorithm has the following iterative form:
\begin{equation}\label{GN_iteration}
x_{k+1}=x_k+\left(J^\top (x_k) J(x_k)\right)^{-1}  J^\top (x_k)(d-\bar f(x_k)),
\end{equation} 
where 
\begin{align*}
d&=[d_1,\ldots,d_m]^\top, \\
\bar f(x_k)&=\left[ f_1(x_k),\ldots,f_m(x_k)\right] ^\top, \\
J(x_k)&=\left[\frac{x_k-a_1}{\|x_k-a_1\|},\ldots,\frac{x_k-a_m}{\|x_k-a_m\|}\right]^\top.
\end{align*}
The initial value of the GN iteration $x_0$ is the consistent estimate obtained in Step 1.
	Specifically, from the definition of $J(x)$ and Assumption~\ref{convergence_of_sample_distribution},
	\begin{equation*}
	\frac{J^\top (x) J(x)}{m} \rightarrow H(x):=\mathbb E_{a \sim \mu} \left[ \frac{(x-a)(x-a)^\top}{\|x-a\|^2}\right]. 
	\end{equation*}
	Given Assumption~\ref{deployment_of_sensor}, $H(x)$ is nonsingular for any $x \in \mathbb R^n$. Therefore, the matrix $J^\top(x)J(x)$ is invertible almost surely in a large sample case.
The standard two-step architecture given above has the attractive property that the resulting two-step estimator has the same asymptotic property as the original LS one~\cite{lehmann2006theory}. Note that the conclusion in~\cite{lehmann2006theory} is for general point estimation problems. Here we rephrase it in the context of our formulation. Before that, we give the definition of stochastic boundedness.
\begin{definition}
	The notation $X_m=O_p(a_m)$ means that the set of values $X_m/a_m$ is stochastically bounded. That is, for any $\epsilon >0$, there exists a finite $M$ and a finite $N$ such that $P\left(|X_m/a_m|>N \right)<\epsilon $ for any $m>M$.
\end{definition}
\begin{theorem}[ {\cite[Theorem 4.3]{lehmann2006theory}}] \label{theorem_two_step}
	Suppose that $\hat x_{m}$ is a $\sqrt{m}$-consistent estimate of $x^o$, i.e., $\hat x_{m}-x^o=O_p(1/\sqrt{m})$. Denote the one-step GN iteration of $\hat x_{m}$ by $\hat x^{\rm GN}_m$, i.e.,
	\begin{equation*}
	\hat x^{\rm GN}_m=\hat x_{m}+\left(J^\top (\hat x_{m}) J(\hat x_{m})\right)^{-1}  J^\top (\hat x_{m})(d-\bar f(\hat x_{m})).
	\end{equation*} 
	Then, under Assumptions~\ref{assumption:Gaussian_noise}-\ref{nonsingular_assumption} or Assumptions~\ref{assumption:Gaussian_noise}-\ref{convergence_of_sample_distribution},\ref{deployment_of_sensor}, we have 
	\begin{equation*}
	\hat x^{\rm GN}_m-\hat x^{\rm LS}_{m}=o_p(1/\sqrt{m}).
	\end{equation*}
	That is $\hat x^{\rm GN}_m$ has the same asymptotic property that $\hat x^{\rm LS}_{m}$ possesses. 
\end{theorem}

\subsection{Global and Consistent Localization} \label{root_m_consistent_localization}

In virtue of Theorem~\ref{theorem_two_step}, the keypoint of the two-step architecture is to construct a $\sqrt{m}$-consistent estimator in the first step. In the following, in terms of whether the variance of measurement noises is known a priori or not, we propose different $\sqrt{m}$-consistent estimators. 

\subsubsection{Consistent Localization with Prior Noise Statistical Knowledge}
With known $\sigma^2$, we can subtract $\sigma^2$ from both sides of~\eqref{squared_range_measurements}:
\begin{equation}\label{squared_range_measurements_unbiased}
d_{i}^2-\sigma^2=\|a_i-x^o\|^2+\varepsilon_{i}, 
\end{equation}
where $\varepsilon_{i}=e_{i}-\sigma^2$ has $0$ mean. Then the bias-eliminate (Bias-Eli) problem is constructed:
\begin{equation}\label{Bias_eliminate_problem}
\hbox{(\textbf{Bias-Eli}):}~~~~\mathop{\rm minimize}_{x\in\mathbb R^{n}} \sum_{i=1}^{m}(\|a_i-x\|^2-d_{i}^2+\sigma^2)^2.
\end{equation}
We call a solution to~\eqref{Bias_eliminate_problem} the Bias-Eli estimate, denoted as $\hat x^{\rm BE}_m$, of $x^o$. Let $y=[x^\top~ \|x\|^2]^\top$. Problem~\eqref{Bias_eliminate_problem} can be equivalently converted into the following constrained problem~\cite{beck2008exact}:
\begin{subequations}\label{GTRS_problem}
	\begin{align}
	\mathop{\rm minimize~}\limits_{y\in\mathbb R^{n+1}} ~& \|Ay-b\|^2 \label{GTRS_objective} \\
	\hbox{\textbf{(Bias-Eli)}:}~~~~\mathop{\rm subject~to~} ~& y^\top D y+2g^\top y=0,\label{GTRS_constraint}
	\end{align}
\end{subequations}
where 
\begin{equation*}
A=\left[ 
\begin{array}{ccc}
-2a_1^\top & 1  \\
\vdots      & \vdots\\
-2a_m^\top & 1
\end{array}
\right] ,
b=\left[ 
\begin{array}{ccc}
d_1^2-\|a_1\|^2-\sigma^2   \\
\vdots      \\
d_m^2-\|a_m\|^2-\sigma^2 
\end{array}
\right]
\end{equation*}
and
\begin{equation*}
D=\left[ 
\begin{array}{ccc}
I_n & 0_{n \times 1}  \\
0_{1 \times n} & 0
\end{array}
\right] ,
g=\left[ 
\begin{array}{ccc}
0_{n \times 1}   \\
-0.5 
\end{array}
\right],
\end{equation*}
where $I_n$ denotes an identity matrix of size $n$ and $0_{i \times j}$ denotes an $i \times j$ matrix whose elements are all $0$. The constraint~\eqref{GTRS_constraint} indicates that $[y]_{n+1}$ equals the square of the norm of $[y]_{1:n}$.
The Bias-Eli problem~\eqref{GTRS_problem} is a nonconvex GTRS problem, which consists of a quadratic objective function and a quadratic equality constraint. A characterization and a rough seeking algorithm of its global minimizer are available in~\cite{more1993generalizations,beck2008exact}. However, the algorithms in~\cite{more1993generalizations,beck2008exact} are not complete. A ``hard case'', where the Hessian matrix of the Lagrangian function is singular, has not been tackled. In addition, some Boolean conditions therein are hard to verify. Here we propose a complete algorithm to seek the solution to the Bias-Eli problem. 
By~\cite[Theorem 3.2]{more1993generalizations}, $y^* \in \mathbb R^{n+1}$ is an optimal solution of~\eqref{GTRS_problem} if and only if there exists a $\lambda^* \in \mathbb R$ such that 
\begin{subequations}
	\begin{align}
	(A^\top A+\lambda^* D)y^* &=A^\top b-\lambda^* g \label{iif_condition_1}\\
	{y^*}^\top D y^*+2g^\top y^* &=0 \label{iif_condition_2}\\
	A^\top A+\lambda^* D & \succeq 0. \label{iif_condition_3}
	\end{align}
\end{subequations}
In what follows, we will illustrate how to seek a pair of $\lambda^*$ and $y^*$ that satisfy~\eqref{iif_condition_1}-\eqref{iif_condition_3}. When $A^\top A+\lambda D$ is nonsingular, we can solve $y$ from~\eqref{iif_condition_1} as
\begin{equation} \label{expression_of_y}
y(\lambda)=(A^\top A+\lambda D)^{-1} (A^\top b-\lambda g).
\end{equation}
Define
\begin{equation} \label{decreasing_function}
c(\lambda):=y(\lambda)^\top D y(\lambda)+2g^\top y(\lambda).
\end{equation}
Based on the structure of $D$, the set of $\lambda$ which satisfies~\eqref{iif_condition_3} has the form $\overline {\mathcal I}= [\lambda_l,\infty)$, where $\lambda_l$ can be obtained via solving the following LMI problem:
\begin{equation} \label{LMI_problem}
\begin{split}
\mathop{\rm minimize}_{\lambda \in\mathbb R}~~ &\lambda \\
{\rm subject~to} ~~& A^\top A+\lambda D \succ 0.
\end{split}
\end{equation}
Define ${\mathcal I}:= (\lambda_l,\infty)$. By~\cite[Theorem 5.2]{more1993generalizations}, $c(\lambda)$ is strictly decreasing over $\mathcal I$. Therefore, when $\lambda^* \in \mathcal I$, a simple bisection method can be applied to search $\lambda^*$~\cite{more1993generalizations,beck2008exact}, i.e., the unique root of $c(\lambda)$. When $A^\top A+\lambda^* D$ is singular, i.e., $\lambda^*=\lambda_l$, which is called the ``hard case'' and omitted in~\cite{more1993generalizations,beck2008exact}, it fails to express $y^*$ by~\eqref{expression_of_y}, and we need to in turn solve the following equations:
\begin{equation} \label{instead_equations}
\begin{cases}
(A^\top A+\lambda_l D)y &=A^\top b-\lambda_l g \\
y^\top D y+2g^\top y &=0.
\end{cases}
\end{equation}
We note that~\eqref{instead_equations} can be solved by a system of polynomial equations rooting. Since $A^\top A+\lambda_l D$ is singular, the solution of~\eqref{instead_equations} may not be unique. Here, we construct a convex optimization problem whose global minimizer is a solution of~\eqref{instead_equations}:
\begin{subequations}\label{convex_problem}
	\begin{align}
	\mathop{\rm minimize~~}\limits_{y\in\mathbb R^{n+1}} ~& v^\top y\label{convex_objective}\hspace{30mm}\\
	\mathop{\rm subject~to~~} ~& (A^\top A+\lambda_l D)y -A^\top b+\lambda_l g =0,\label{affine_equality}\\
	& y^\top D y+2g^\top y  \leq 0,\label{convex_inequality}
	\end{align}
\end{subequations}
where $v \neq 0$ is an eigenvector of $A^\top A+\lambda_l D$ associated with eigenvalue $0$.
It can be readily verified that~\eqref{convex_problem} is a convex optimization problem. Now we show that ``='' in~\eqref{convex_inequality} can be achieved at any minimizer. Suppose there exists an optimal $y'$ satisfying~\eqref{affine_equality} and $y'^\top D y'+2g^\top y'  < 0$. Note that $y'^\top D y'+2g^\top y'  < 0$, there exists a $\delta>0$ such that $y''^\top D y''+2g^\top y''  \leq 0$, where $y''=y'-\delta v$. Since $v$ is an eigenvector of $A^\top A+\lambda_l D$ associated with eigenvalue $0$, $y''$ satisfies~\eqref{affine_equality} automatically. Moreover, we have $v^\top y''=v^\top y'-\delta v^\top v<v^\top y'$, which contradicts the optimality of $y'$. Therefore, the global minimizers of~\eqref{convex_problem} are solutions of~\eqref{instead_equations}. 

By now, we have introduced the methods to calculate the Bias-Eli solution when the matrix $A^\top A+\lambda^* D$ is nonsingular and singular, respectively. However, it is still unclear that how to judge whether $A^\top A+\lambda^* D$ is singular or not. Next, we will tackle this problem and give the whole programmable algorithm to calculate the Bias-Eli solution. 
Since $A^\top A$ is positive definite, it follows that $A^\top A$ and $D$ can be simultaneously diagonalized~\cite{zeng2022localizability}, i.e., there exists a nonsingular matrix $R$ for which 
\begin{equation} \label{simultaneously_diagonalizing}
\begin{split}
R^\top A^\top AR &=\Gamma:={\rm diag} (\gamma_1,\ldots,\gamma_{n+1}),\\
R^\top DR &=\Sigma:={\rm diag} ({\delta}_1,\ldots,{\delta}_{n+1}).
\end{split}
\end{equation}
Further, following the similar derivation with Lemma 2 in~\cite{zeng2022localizability}, finding all $\lambda \in \mathcal I$ that satisfy $c(\lambda)=0$ is equivalent to finding the roots (located in $\mathcal I$) of the following $(2n+2)$-order polynomial:
\begin{align} 
T(\lambda):=&\sum\limits_{i = 1}^{n + 1} \left[w(\lambda) \right]_i \left[2g^\top R \right]_i (\gamma_i + \lambda {\delta _i}) \prod\limits_{j = 1,j \ne i}^{n + 1} {{{\left( {\gamma_j + \lambda {\delta _j}} \right)}^2}} \nonumber\\
& +\sum\limits_{i = 1}^{n + 1} \left[w(\lambda) \right]_i^2 \delta_i \prod\limits_{j = 1,j \ne i}^{n + 1} {{{\left( {\gamma_j + \lambda {\delta _j}} \right)}^2}}, \label{polynomial}
\end{align}
where $w(\lambda)=R^\top (A^\top b-\lambda g)$.

Then by using the Sturm's theorem~\cite{thomas1941sturm}, we can judge whether $T(\lambda)$ has a root in $\mathcal I$. If it is the case, we can calculate the upper bound of the root in virtue of the Cauchy's bound~\cite{rahman2002analytic} and apply a simple bisection method to seek $\lambda^*$. If $T(\lambda)$ has no roots in $\mathcal I$, i.e., $A^\top A+\lambda^* D$ is singular, we resort to solving the convex problem~\eqref{convex_problem}. For real applications, each sensor measures the distances of the target and send the measurements to a base station where the Bias-Eli problem~\eqref{GTRS_problem} is constructed and solved. The whole algorithm which solves the Bias-Eli problem~\eqref{GTRS_problem} is presented in Algorithm~\ref{pseudo_algorithm_Bias_Eli}.
\begin{algorithm}
	\caption{Solution Seeking for the Bias-Eli Problem~\eqref{GTRS_problem}}
	\label{pseudo_algorithm_Bias_Eli}
	\begin{algorithmic}[1]
		\State Input $A$, $b$, $D$, and $g$.
		\State Calculate $\lambda_l$ by solving LMI problem~\eqref{LMI_problem}.
		\State Construct $T(\lambda)$ according to (\ref{polynomial}).
		\State {\bf IF}  $T(\lambda)=0$ has a root in $\mathcal I$ (By Sturm's theorem):
		\State ~~~~Calculate the upper bound of the root of $T(\lambda)=0$ (By Cauchy's bound).
		\State ~~~~Use a bisection method to obtain $\hat \lambda^*$.
		\State ~~~~Calculate $\hat y^{\rm BE}_m=(A^\top A+\hat \lambda^* D)^{-1} (A^\top b-\hat \lambda^* g)$.
		\State {\bf ELSE}: 
		\State ~~~~Calculate $\hat y^{\rm BE}_m$ by solving convex problem~\eqref{convex_problem}. 
		\State {\bf END IF}
		\State  Output $\hat y^{\rm BE}_m$.
	\end{algorithmic}
\end{algorithm}

In what follows, we investigate the asymptotic property of
the Bias-Eli estimator $\hat x^{\rm BE}_m$. The following theorem shows the consistency of $\hat x^{\rm BE}_m$.
\begin{theorem} \label{consistency_of_bias_eli}
	The Bias-Eli estimator is $\sqrt{m}$-consistent, i.e., $\hat x^{\rm BE}_m - x^o=O_p(1/\sqrt{m})$.
\end{theorem}
The proof of Theorem~\ref{consistency_of_bias_eli} is presented in Appendix~\ref{proof_of_BE_consistency}.

\begin{remark}
	The asymptotic unbiasedness of the Bias-Eli estimator is due to the $0$ mean of error terms $\varepsilon_i$'s. We remark that the error terms $e_i$'s in the S-LS problem~\eqref{S-LS} are not mean $0$, which leads to the biasedness of the S-LS estimator, even in the asymptotic case. Therefore, the S-LS estimate is not consistent, which will be shown in Fig.~\ref{asymptotic_rmse} of our simulations.
\end{remark}

Note that solving the Bias-Eli problem~\eqref{GTRS_problem} involves massive numerical computing for the solution to a nonconvex optimization problem, which is computationally inefficient. In this part, we focus on an approximation of~\eqref{GTRS_problem} by discarding the quadratic constraint~\eqref{GTRS_constraint} and show that the resultant ordinary least squares estimator is consistent. By discarding~\eqref{GTRS_constraint}, we obtain the following linear least squares problem:
\begin{equation} \label{linear_Bias_Eli}
\hbox{\textbf{(Bias-Eli-Lin)}:}~~~~~\mathop{\rm minimize~}\limits_{y\in\mathbb R^{n+1}} ~ \|Ay-b\|^2 .
\end{equation}
The optimal solution to~\eqref{linear_Bias_Eli} has the closed form $\hat y^{\rm BEL}_m=(A^\top A)^{-1} A^\top b$. Theorem~\ref{consistency_of_linear_bias_eli_estimate} gives the consistency of the Bias-Eli-Lin estimator $\hat x^{\rm BEL}_m=\left[\hat y^{\rm BEL}_m \right]_{1:n} $. 
\begin{theorem} \label{consistency_of_linear_bias_eli_estimate}
	The Bias-Eli-Lin estimator is $\sqrt{m}$-consistent, i.e., $\hat x^{\rm BEL}_m - x^o=O_p(1/\sqrt{m})$.
\end{theorem}
The proof of Theorem~\ref{consistency_of_linear_bias_eli_estimate} is presented in Appendix~\ref{proof_of_BEL_consistency}.

\subsubsection{Simultaneous Localization and Noise Statistical Inference}
When the noise variance $\sigma^2$ is unknown, we are going to simultaneously estimate $x^o$ and $\sigma^2$. We replace $\sigma^2$ as $c$ and modify~\eqref{Bias_eliminate_problem} to the following problem:
\begin{subequations}\label{Noise_Estimate_problem}
	\begin{align}
	\mathop{\rm minimize}_{x\in\mathbb R^{n}, c \in \mathbb R}~ &\sum_{i=1}^{m}(\|a_i-x\|^2-d_{i}^2+c)^2 \label{Noise_Est_obj}\\
	\hbox{\textbf{(Noise-Est)}:}~~~\mathop{\rm subject~to} ~& c \geq 0.\label{Noise_Est_constraint}
	\end{align}
\end{subequations}
Let $\bar y=[x^\top ~\|x\|^2+c]^\top$, Problem~\eqref{Noise_Estimate_problem} can be equivalently converted into 
\begin{subequations}\label{Noise_Estimate_problem2}
	\begin{align}
	\mathop{\rm minimize}_{\bar y\in\mathbb R^{n+1}} ~&\|A \bar y-\bar b\|^2 \label{Noise_Est_obj2}\\
	\hbox{\textbf{(Noise-Est)}:}~~~~\mathop{\rm subject~to} ~& \bar y^\top D \bar y+2g^\top \bar y \leq 0,\label{Noise_Est_constraint2}
	\end{align}
\end{subequations}
where
\begin{equation*}
\bar b=\left[ 
\begin{array}{cc}
d_1^2-\|a_1\|^2  \\
\vdots \\
d_m^2-\|a_m\|^2
\end{array}
\right] ,
\end{equation*}
and $A$, $D$, and $g$ are the same as that in~\eqref{GTRS_problem}.
Since the objective function~\eqref{Noise_Est_obj2} and the inequality constraint~\eqref{Noise_Est_constraint2} are both quadratic functions with a positive semidefinite Hessian matrix, Problem~\eqref{Noise_Estimate_problem2} is a convex one, which can be efficiently solved by mature first and second-order algorithms. We call a solution to~\eqref{Noise_Estimate_problem2} the Noise-Est estimator, denoted as $\hat y^{\rm NE}_m$. Then the Noise-Est estimates of object position $\hat x^{\rm NE}_m$ and noise variance $\hat \sigma^{\rm NE}_m$ are obtained:
\begin{subequations}
	\begin{align}
	\hat x^{\rm NE}_m&=\left[ \hat y^{\rm NE}_m\right] _{1:n} \label{NE_target1} \\
	\hat \sigma^{\rm NE}_m&=\left[ \hat y^{\rm NE}_m\right] _{n+1}-\|\left[ \hat y^{\rm NE}_m\right] _{1:n}\|^2. \label{NE_noise1}
	\end{align} 
\end{subequations}
Let $\bar y^o=\left[{x^o}^\top ~\|x^o\|^2+\sigma^2 \right] ^\top$. The following theorem gives the consistency of $\hat y^{\rm NE}_m$.
\begin{theorem} \label{consistency_of_Noise_Est}
	The Noise-Est estimator is $\sqrt{m}$-consistent, i.e., $\hat y^{\rm NE}_m - \bar y^o=O_p(1/\sqrt{m})$.
\end{theorem}
The proof of Theorem~\ref{consistency_of_Noise_Est} is presented in Appendix~\ref{proof_of_NE_consistency}.

In the following, we consider the unconstrained version of~\eqref{Noise_Estimate_problem2}. By relaxing the constraint~\eqref{Noise_Est_constraint2}, we obtain the following ordinary (linear) least squares problem:
\begin{equation}\label{Noise_estimate_linear}
\hbox{(\textbf{Noise-Est-Lin}):}~~~~\mathop{\rm minimize}_{\bar y\in\mathbb R^{n+1}} \|A \bar y-\bar b\|^2.
\end{equation}
Let $\hat y^{\rm NEL}_m$ denote the optimal solution to problem~\eqref{Noise_estimate_linear}; it has the closed-form $\hat y^{\rm NEL}_m=(A^\top A)^{-1} A^\top \bar b$. Then the Noise-Est-Lin estimates of object position $\hat x^{\rm NEL}_m$ and noise variance $\hat \sigma^{\rm NEL}_m$ are obtained:
\begin{subequations}
	\begin{align}
	\hat x^{\rm NEL}_m&=\left[ \hat y^{\rm NEL}_m\right] _{1:n} \label{NE_target} \\
	\hat \sigma^{\rm NEL}_m&=\left[ \hat y^{\rm NEL}_m\right] _{n+1}-\|\left[ \hat y^{\rm NEL}_m\right] _{1:n}\|^2. \label{NE_noise}
	\end{align} 
\end{subequations}
Note that $\hat y^{\rm NEL}_m$ is identical to $\hat y^{\rm NE0}_m$ in the proof of Theorem~\ref{consistency_of_Noise_Est}; see~\eqref{limit_probability} for the definition of $\hat y^{\rm NE0}_m$. Since we have shown $\hat y^{\rm NE0}_m$ is $\sqrt{m}$-consistent in the proof of Theorem~\ref{consistency_of_Noise_Est}, the following theorem holds straightforwardly.
\begin{theorem} \label{consistency_of_noise_estimate}
	The Noise-Est-Lin estimator is $\sqrt{m}$-consistent, i.e., $\hat y^{\rm NEL}_m - \bar y^o=O_p(1/\sqrt{m})$.
\end{theorem}

\subsection{Heterogeneity of Noise Variances}
Different sensors may have distinct measurement variances $\sigma^2_i,i=1,\ldots,m$. In this case, with known noise variances, we can modify the measurement model~\eqref{range_measurement_model} as follows:
\begin{equation}\label{general_measurement_model}
\bar d_{i}=\|a_i-x^o\|/\sigma_i+\bar r_{i},
\end{equation}
where $\bar d_{i}=d_{i}/\sigma_i$ and $\bar r_{i}=r_{i}/\sigma_i$. Since the modified noise terms $\bar r_i$ have the same variance, the solution to the following weighted least squares (WLS) problem is consistent and asymptotically efficient:
\begin{equation}\label{WLS}
\hbox{(\textbf{WLS}):}~~~~\mathop{\rm minimize}_{x\in\mathbb R^n} \sum_{i=1}^{m}  \frac{\left(\|a_i-x\|-d_{i}\right)^2}{\sigma^2_i}.
\end{equation}
Similarly, we can use a two-step scheme to develop an estimator that is strongly consistent and asymptotically efficient. Specifically, in the first step, the following weighted Bias-Eli-Lin problem needs to be solved:
\begin{equation} \label{R_linear_Bias_Eli}
\hbox{\textbf{(W-Bias-Eli-Lin)}:}~~~\mathop{\rm minimize~}\limits_{y\in\mathbb R^{n+1}} ~ (Ay-b_{\sigma})^\top W (Ay-b_{\sigma}),
\end{equation}
where 
$$
b_{\sigma}=\left[ 
\begin{array}{ccc}
d_1^2-\|a_1\|^2-\sigma^2_1   \\
\vdots      \\
d_m^2-\|a_m\|^2-\sigma^2_m 
\end{array}
\right],
$$
and $W={\rm diag}\left( \frac{1}{\sigma_1^2},\ldots,\frac{1}{\sigma_m^2}\right)$ is the weight matrix. 
In the second step, a Gauss-Newton iteration associated with~\eqref{general_measurement_model} is executed. The result is verified in Fig.~\ref{asymptotic_rmse_WBEL} in Section~\ref{simulations}.

We next consider the case where $M$ sensors have unknown and distinct measurement noise variances $\sigma^2_i,i=1,\ldots,M$. Since the number of variables to be estimated exceeds the number of sensors, it is impossible to simultaneously estimate the target position and all noise variance like~\eqref{Noise_estimate_linear}. Instead, we can first design a consistent estimate of each sensor's variance $\hat \sigma^2_i$ based on repetitive measurements as follows~\cite{lehmann2006theory}:
\begin{equation} \label{consistent_estimate_variance}
\hat \sigma^2_i=\frac{1}{T-1} \sum_{j=1}^{T} (d_{ij}-\bar d_i)^2,
\end{equation}
where $d_{ij}$ is the $j$-th measurement of the $i$-th sensor, and $\bar d_i=\sum_{j=1}^{T} d_{ij} /T$ with $T$ being the number of repetitive measurements by each sensor. There are totally $m=MT$ measurements, and we construct the following approximate weighted least squares (AWLS) problem:
\begin{equation}\label{AWLS}
		\hbox{(\textbf{AWLS}):}~~~~\mathop{\rm minimize}_{x\in\mathbb R^n} \sum_{i=1}^{M}\sum_{j=1}^{T}  \frac{\left(\|a_i-x\|-d_{ij}\right)^2}{\hat \sigma^2_i}.
\end{equation}
Similarly, we first solve the following approximate weighted Bias-Eli-Lin problem:
\begin{equation} \label{AR_linear_Bias_Eli}
\hbox{\textbf{(AW-Bias-Eli-Lin)}:}\mathop{\rm minimize~}\limits_{y\in\mathbb R^{n+1}} ~ (Ay-\hat b_{\sigma})^\top \hat W (Ay-\hat b_{\sigma}),
\end{equation}
where 
$$
\hat b_{\sigma}=\left[ 
\begin{array}{ccc}
d_{11}^2-\|a_1\|^2-\hat \sigma^2_1   \\
\vdots      \\
d_{MT}^2-\|a_M\|^2-\hat \sigma^2_M
\end{array}
\right] \in \mathbb R^{m},
$$
and $\hat W={\rm diag}\left( \frac{1}{\hat \sigma_1^2},\ldots,\frac{1}{\hat \sigma_1^2},\ldots,\frac{1}{\hat \sigma_M^2},\ldots,\frac{1}{\hat \sigma_M^2}\right)$ is the weight matrix. Note that 
\begin{align*}
& (A^\top \hat W A)^{-1} A^\top \hat W \hat b_{\sigma}-(A^\top W A)^{-1} A^\top W b_{\sigma} \\
= & \left( (A^\top \hat W A)^{-1}-(A^\top W A)^{-1}\right) A^\top \hat W \hat b_{\sigma} \\
& +( A^\top W A)^{-1} A^\top ( \hat W \hat b_{\sigma}-W b_{\sigma})  \\
= & O_p(1/\sqrt{m}) ( A^\top W A)^{-1} A^\top W b_{\sigma}(1+O_p(1/\sqrt{m})) \\
& + ( A^\top W A)^{-1} A^\top O_p(1/\sqrt{m})  W b_{\sigma} \\
= & O_p(1/\sqrt{m}) ( A^\top W A)^{-1} A^\top W b_{\sigma},
\end{align*}
which implies that the AW-Bias-Eli-Lin estimator is $\sqrt{m}$-consistent. In the second step, a Gauss-Newton iteration associated with~\eqref{AWLS} is executed.

\section{Finite-sample Analysis for Practical Localization} \label{finite_sample_section}
In this section, we conduct the finite sample analysis of the Bias-Eli-Lin and Noise-Est-Lin estimators in virtue of their closed-form solutions. Specifically, we will derive their biases and MSEs respectively when $m$ is a finite number.

\begin{figure*}[!t]
	\centering
	\begin{subfigure}[b]{0.45\textwidth}
		\centering
		\includegraphics[width=\textwidth]{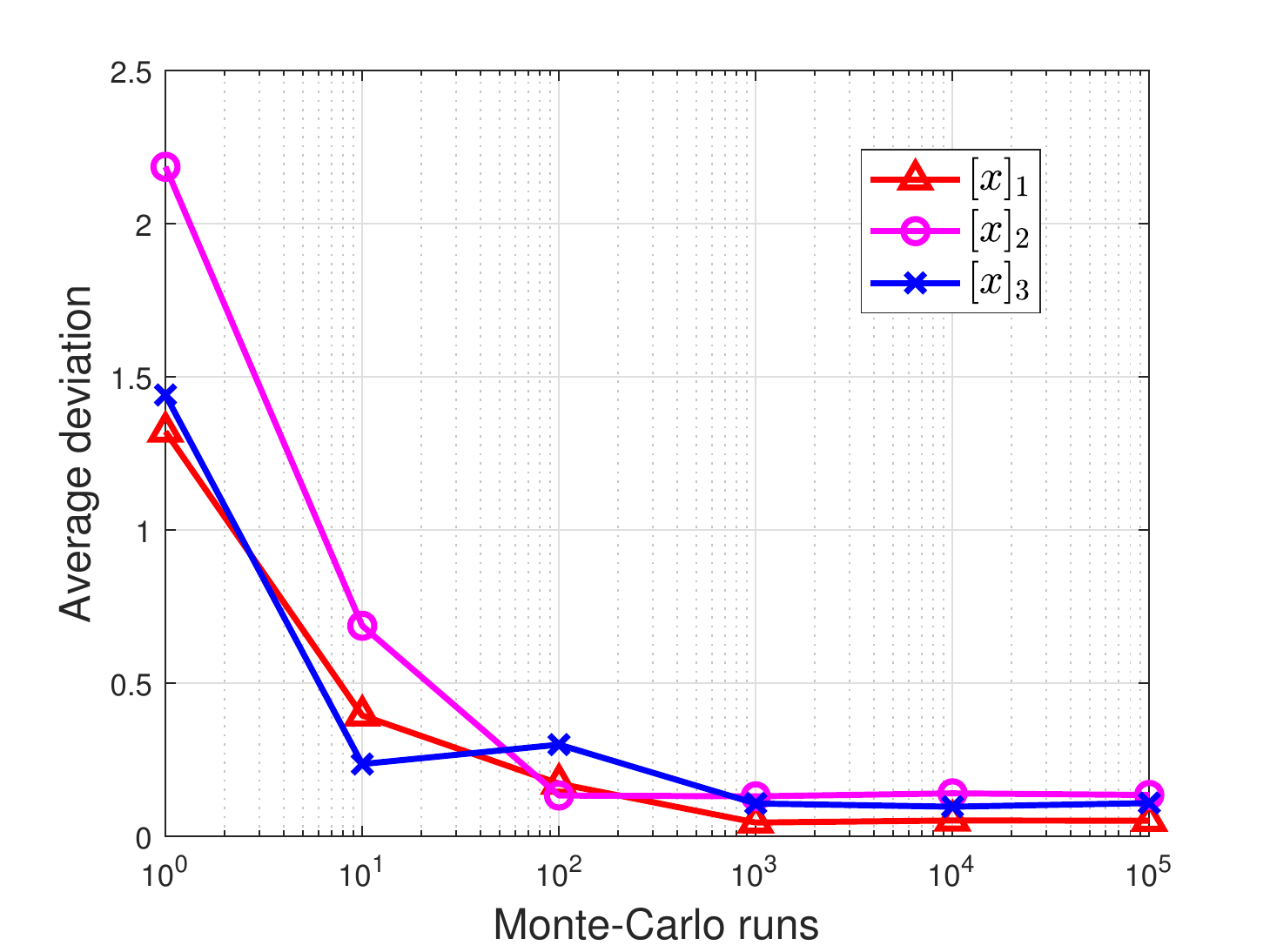}
		\caption{Bias-Eli}
		\label{finite:Bias_Bias_Eli}
	\end{subfigure}
	\hfill
	\begin{subfigure}[b]{0.45\textwidth}
		\centering
		\includegraphics[width=\textwidth]{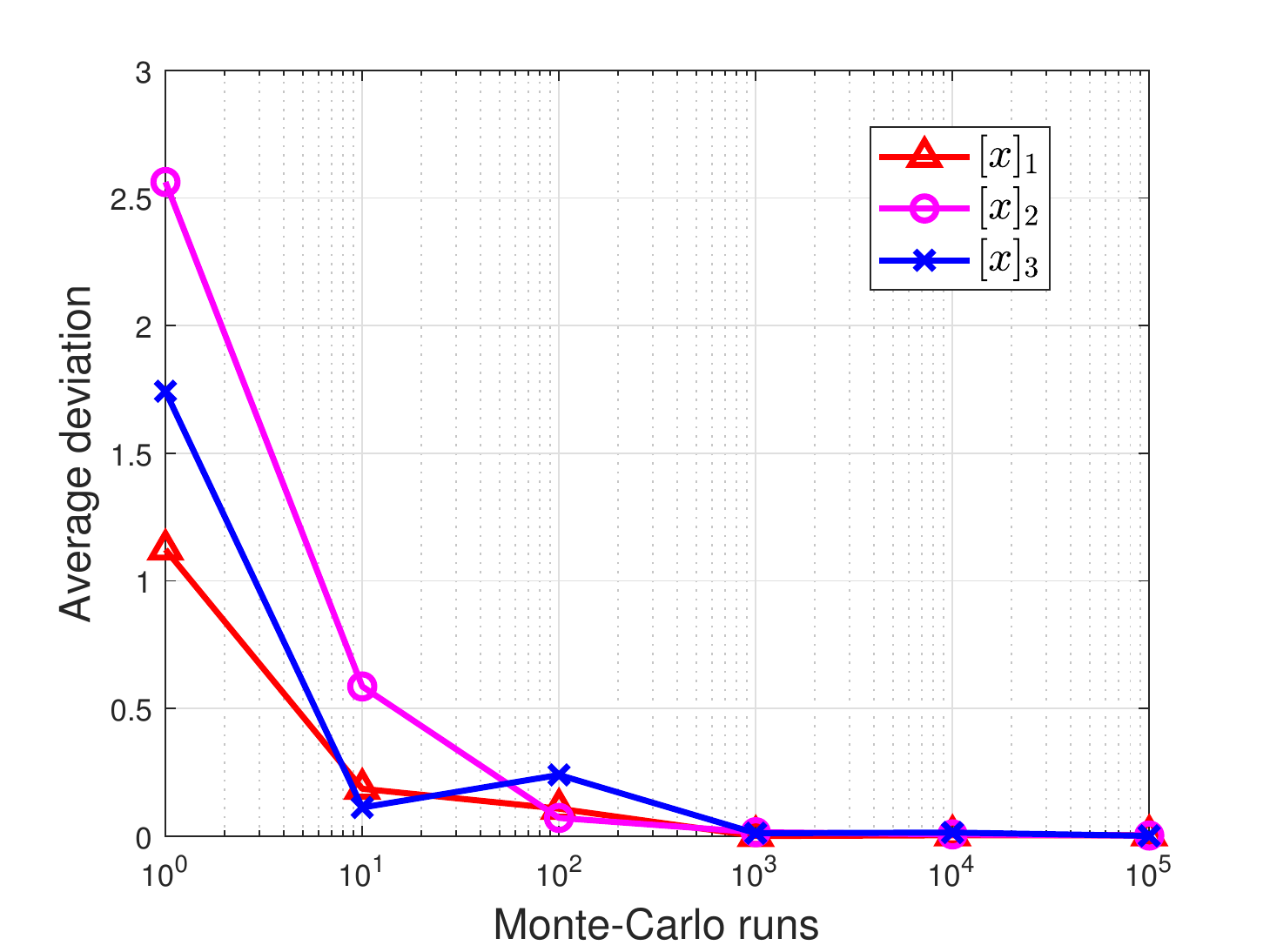}
		\caption{Bias-Eli-Lin}
		\label{finite:Bias_Bias_Eli_Lin}
	\end{subfigure}
	\begin{subfigure}[b]{0.45\textwidth}
		\centering
		\includegraphics[width=\textwidth]{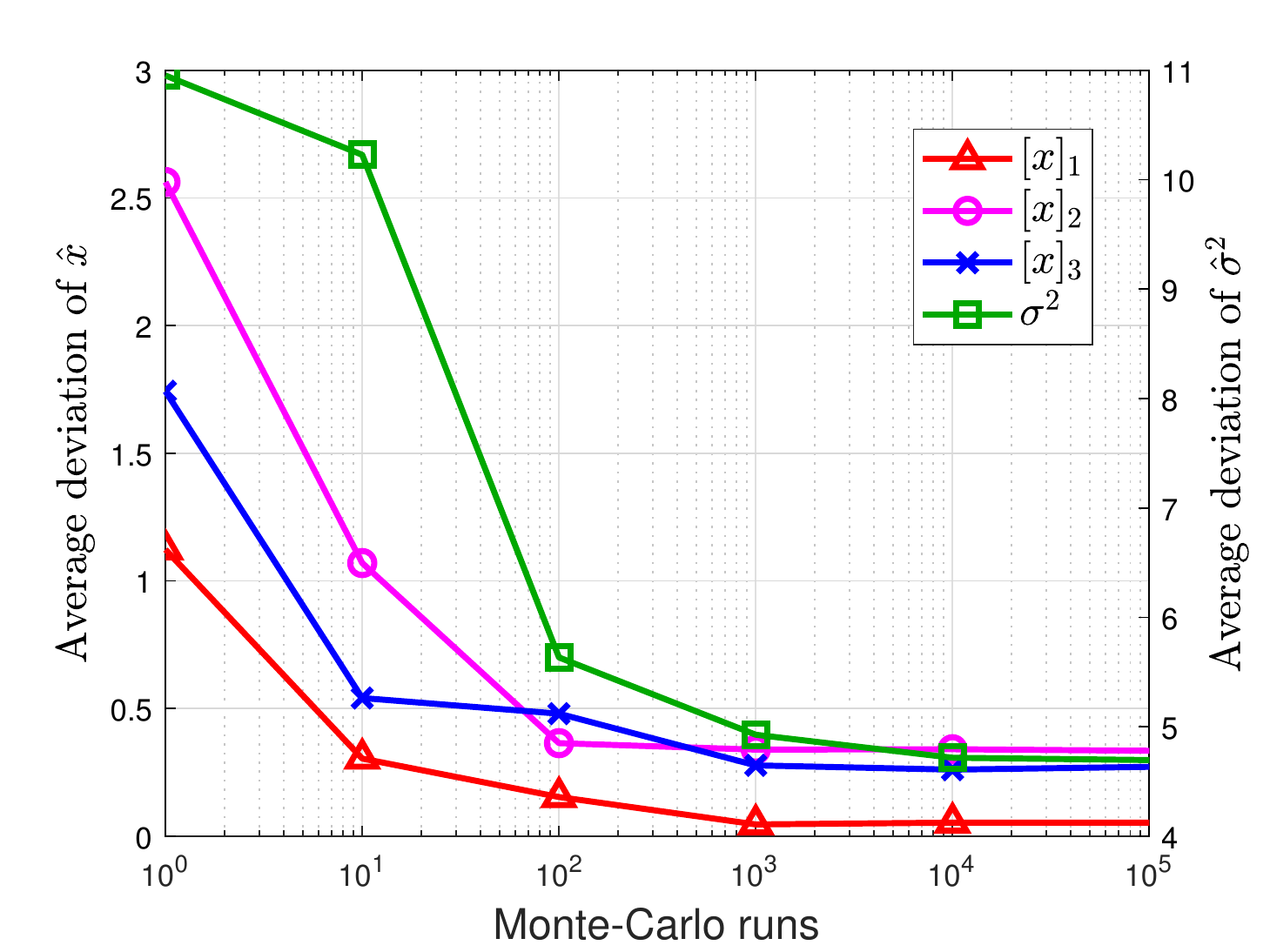}
		\caption{Noise-Est}
		\label{finite:Bias_Noise_Est}
	\end{subfigure}
	\hfill
	\begin{subfigure}[b]{0.45\textwidth}
		\centering
		\includegraphics[width=\textwidth]{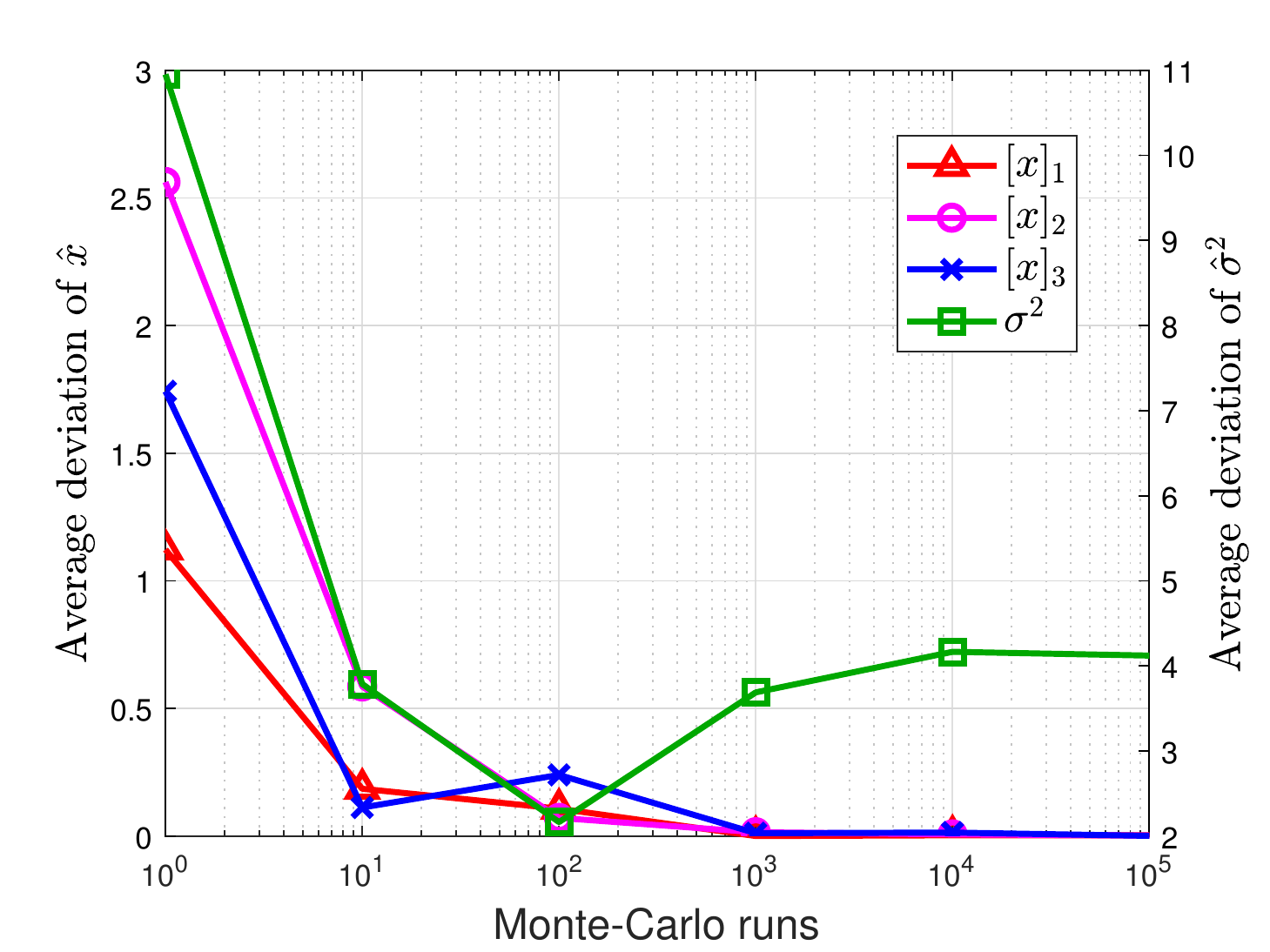}
		\caption{Noise-Est-Lin}
		\label{finite:Bias_Noise_Est_Lin}
	\end{subfigure}
	\caption{Biases of first-step estimators in the finite sample case.}
	\label{finite:Bias}
\end{figure*}

\subsection{Bias and MSE of the Noise-Est-Lin estimator} \label{finite_sample_noise_est}
Recall that the Noise-Est-Lin estimation is the optimal solution of~\eqref{Noise_estimate_linear}, which has the closed form $\hat y^{\rm NEL}_m=(A^\top A)^{-1} A^\top \bar b$. Then we have 
\begin{align*}
\mathbb E \left[ \hat y^{\rm NEL}_m\right] &=(A^\top A)^{-1} A^\top \mathbb E[\bar b] \\
&=(A^\top A)^{-1} A^\top \left[ 
\begin{array}{cc}
-2a_1^\top x^o+\|x^o\|^2+\sigma^2  \\
\vdots \\
-2a_m^\top x^o+\|x^o\|^2+\sigma^2
\end{array}
\right] \\
& =\left[ 
\begin{array}{cc}
x^o  \\
\|x^o\|^2+\sigma^2
\end{array}
\right],
\end{align*}
which indicates that the Noise-Est-Lin estimator is unbiased for the estimation of $x^o$. Now we show that the estimate of $\sigma^2$ is generally biased. Note that
\begin{equation*}
\hat y^{\rm NEL}_m=\left[ 
\begin{array}{cc}
x^o  \\
\|x^o\|^2+\sigma^2
\end{array}
\right]+(A^\top A)^{-1} A^\top \left[ 
\begin{array}{cc}
\varepsilon_1  \\
\vdots \\
\varepsilon_m
\end{array}
\right].
\end{equation*}
Let $\bar A=(A^\top A)^{-1} A^\top$, $\epsilon_i=\sum_{j=1}^{m} \bar A_{ij} \varepsilon_j$, and $\epsilon=[\epsilon_1,\ldots,\epsilon_{n+1}]^\top $. We have 
\begin{equation*}
\hat y^{\rm NEL}_m=\left[ 
\begin{array}{cc}
x^o  \\
\|x^o\|^2+\sigma^2
\end{array}
\right]+ \epsilon.
\end{equation*}
From~\eqref{NE_noise}, we have 
\begin{equation} \label{expression_of_sigma_estimation} 
\begin{split}
\hat \sigma^{\rm NEL}_m &=\|x^o\|^2+\sigma^2+\epsilon_{n+1}-\|x^o+[\epsilon]_{1:n}\|^2 \\
& =   \sigma^2+\epsilon_{n+1}-2{x^o}^\top [\epsilon]_{1:n} -\sum_{i=1}^{n} [\epsilon]_{i}^2 .
\end{split}
\end{equation}
Note that $\mathbb E\left[ \bar A_{ij} \varepsilon_j \bar A_{ik} \varepsilon_k\right] =0$ for $j \neq k$. We obtain 
\begin{align*}
\mathbb E \left[ [\epsilon]_{i}^2\right] &= \sum_{j=1}^{m} {\bar A_{ij}}^2 \mathbb E\left[ \varepsilon_j^2\right] \\
&= \sum_{j=1}^{m} {\bar A_{ij}}^2 \left(  4 f_j^2(x^o) \sigma^2+2 \sigma^2\right), 
\end{align*}
which together with $\mathbb E[\epsilon_i]=0$ for any $i$ give
\begin{equation}\label{bias_sigma}
\mathbb E\left[\hat \sigma^{\rm NEL}_m-\sigma^2 \right]=-\sum_{i=1}^{n}\sum_{j=1}^{m}  \bar A_{ij}^2\left(4f_j^2(x^o) \sigma^2+2\sigma^4\right) .
\end{equation}
Recall that $\bar y^o=\left[{x^o}^\top ~\|x^o\|^2+\sigma^2 \right]^\top $. To obtain the MSE of $\hat y^{\rm NEL}_m$, we first calculate 
\begin{equation} \label{expectation_bb}
\begin{split}
\mathbb E\left[\bar b {{}\bar b}^\top \right] & =\mathbb E\left[ (A \bar y^o+\varepsilon )(A \bar y^o+\varepsilon )^\top\right] \\
&=A \bar y^o {{}\bar y^o}^\top A^\top + \mathbb E[\varepsilon \varepsilon^\top] \\
& = A \bar y^o {{}\bar y^o}^\top A^\top + \Lambda ,
\end{split}
\end{equation}
where $\Lambda$ is shown in~\eqref{covariance_of_noises}. Then we obtain 
\begin{align} 
\mathbb {MSE}(\hat y^{\rm NEL}_m)& =\mathbb E\left[ (\hat y^{\rm NEL}_m-\bar y^o)(\hat y^{\rm NEL}_m-\bar y^o)^\top\right] \nonumber\\
& =\mathbb E\left[\hat y^{\rm NEL}_m {{{}\hat y}^{\rm NEL}_m}^\top \right]- \bar y^o{{}\bar y^o}^\top \nonumber\\
&=(A^\top A)^{-1} A^\top \mathbb E\left[\bar b {{}\bar b}^\top \right] A (A^\top A)^{-1} - \bar y^o{{}\bar y^o}^\top \nonumber\\
&=  (A^\top A)^{-1} A^\top \Lambda A (A^\top A)^{-1}. \label{finite_MSE_NE}
\end{align}

\subsection{Bias and MSE of the Bias-Eli-Lin estimator} \label{finite_sample_bias_eli}
Deriving the bias and MSE for the Bias-Eli-Lin estimator can utilize the results of the Noise-Est-Lin one. Specifically, 
\begin{equation}
\begin{split}
\hat y^{\rm BEL}_m &=(A^\top A)^{-1} A^\top b \\
&= (A^\top A)^{-1} A^\top \bar b -(A^\top A)^{-1} A \left[ 
\begin{array}{cc}
\sigma^2  \\
\vdots \\
\sigma^2
\end{array}
\right] \\
&= \hat y^{\rm NEL}_m -\left[ 
\begin{array}{cc}
0_{n \times 1}  \\
\sigma^2
\end{array}
\right].
\end{split}
\end{equation}
By further noting that $\hat x^{\rm BEL}_m=[\hat y^{\rm BEL}_m]_{1:n}$, and $\hat x^{\rm NEL}_m=[\hat y^{\rm NEL}_m]_{1:n}$, we have $\hat x^{\rm BEL}_m=\hat x^{\rm NEL}_m$. Denote $[{x^o}^\top ~\|x^o\|^2]^\top$ as ${y^o}$. Recall that $\bar y^o=[{x^o}^\top ~\|x^o\|^2+\sigma^2]^\top$. We have 
\begin{equation}
\mathbb E\left[\hat y^{\rm BEL}_m -y^o\right] =\mathbb E\left[\hat y^{\rm NEL}_m -\bar y^o\right]=0 
\end{equation}
and 
\begin{equation} \label{finite_MSE_BEL}
\begin{split}
\mathbb {MSE}(\hat y^{\rm BEL}_m)= & \mathbb E\left[ (\hat y^{\rm BEL}_m- y^o)(\hat y^{\rm BEL}_m-y^o)^\top\right] \\
=& \mathbb E\left[ (\hat y^{\rm NEL}_m- \bar y^o)(\hat y^{\rm NEL}_m-\bar y^o)^\top\right] \\
= & (A^\top A)^{-1} A^\top \Lambda A (A^\top A)^{-1},
\end{split}
\end{equation}
i.e., the Bias-Eli-Lin estimator has the same bias and MSE as the Noise-Est-Lin one. 

\begin{figure*}[t]
	\centering
	\begin{subfigure}[b]{0.45\textwidth}
		\centering
		\includegraphics[width=\textwidth]{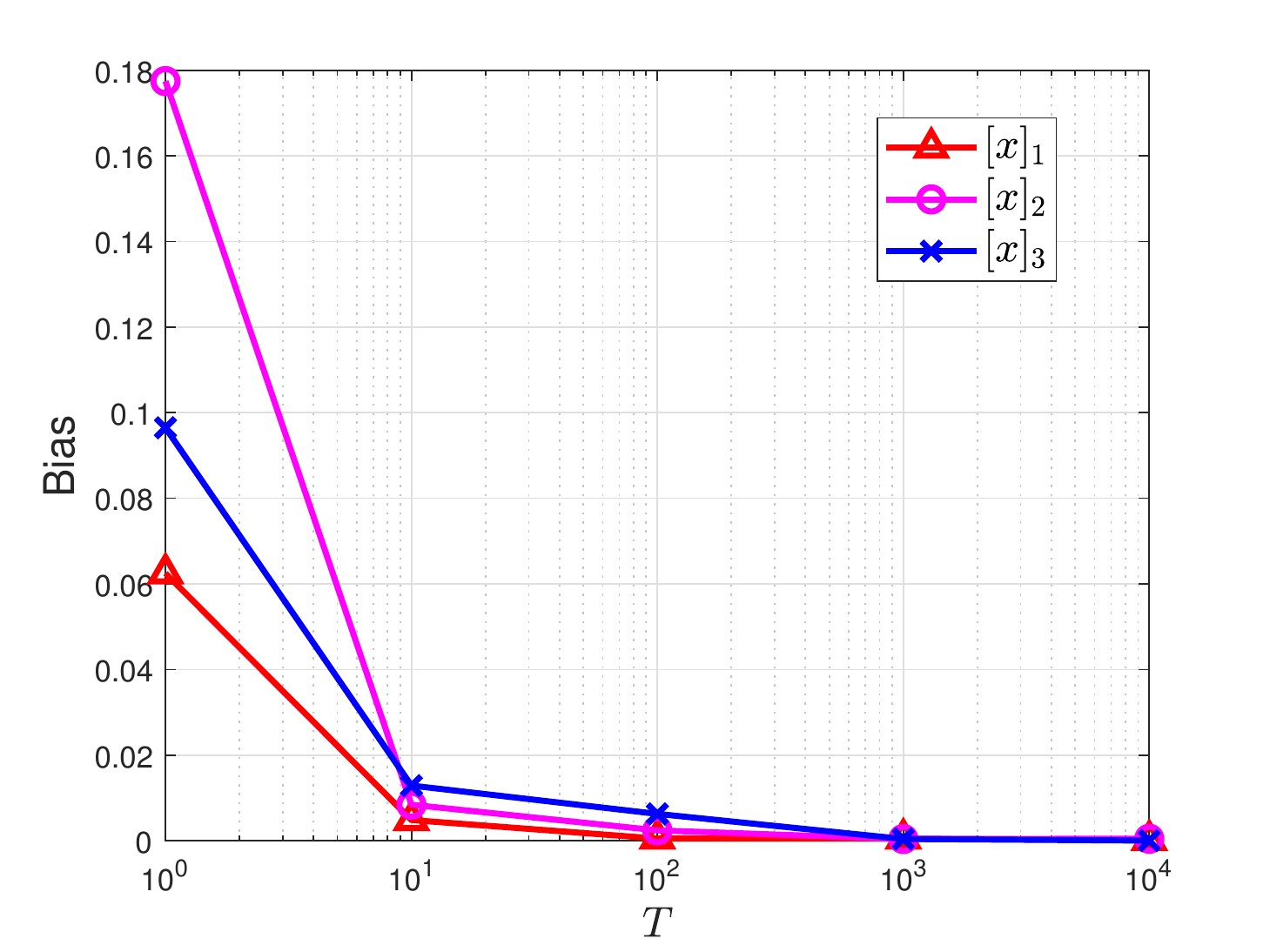}
		\caption{Bias-Eli}
		\label{asymptotic:Bias_Bias_Eli}
	\end{subfigure}
	\hfill
	\begin{subfigure}[b]{0.45\textwidth}
		\centering
		\includegraphics[width=\textwidth]{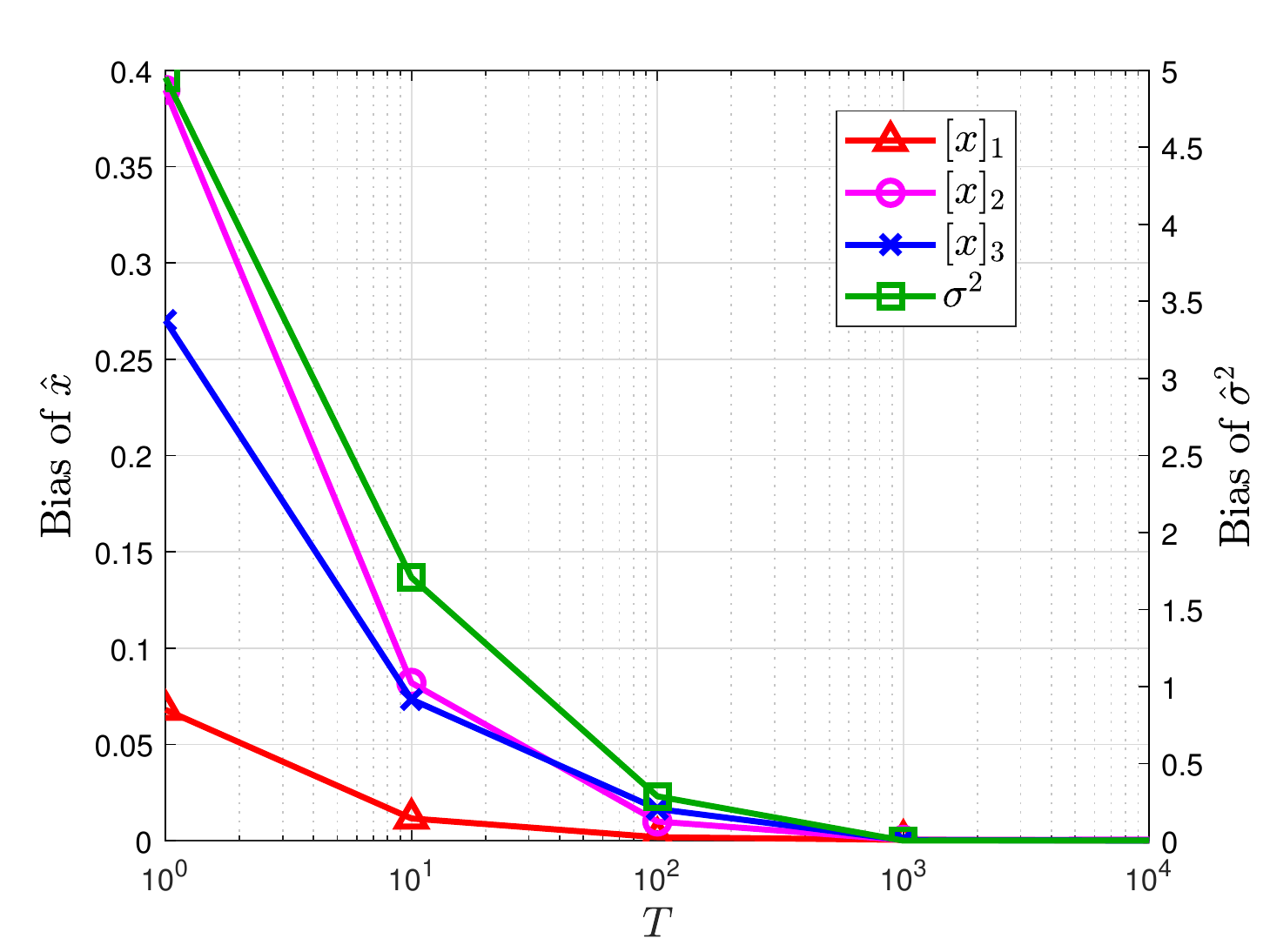}
		\caption{Noise-Est}
		\label{asymptotic:Bias_Noise_Est}
	\end{subfigure}
	\begin{subfigure}[b]{0.45\textwidth}
		\centering
		\includegraphics[width=\textwidth]{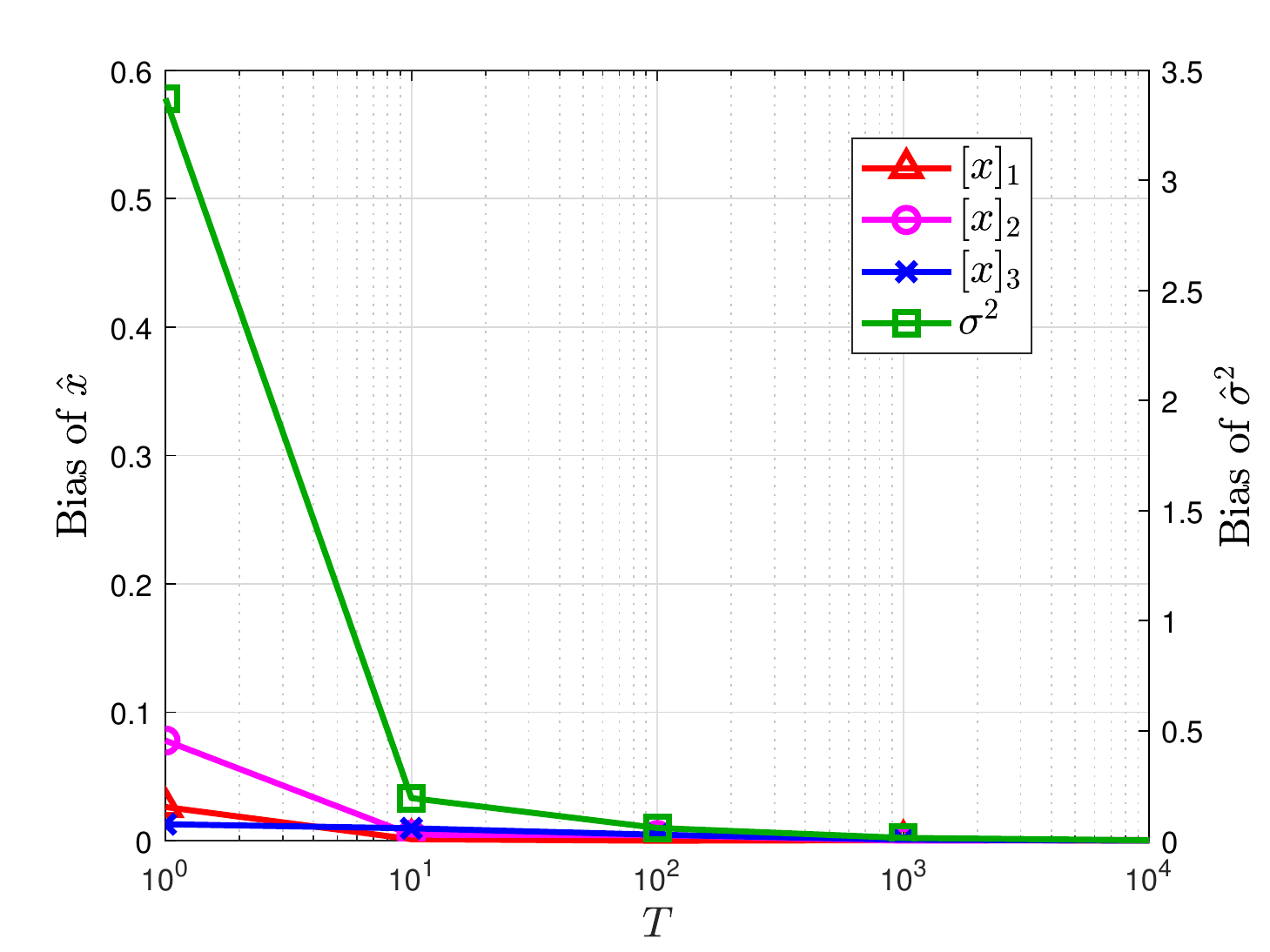}
		\caption{Noise-Est-Lin}
		\label{asymptotic:Bias_Noise_Est_Lin}
	\end{subfigure}
	\hfill
	\begin{subfigure}[b]{0.45\textwidth}
		\centering
		\includegraphics[width=\textwidth]{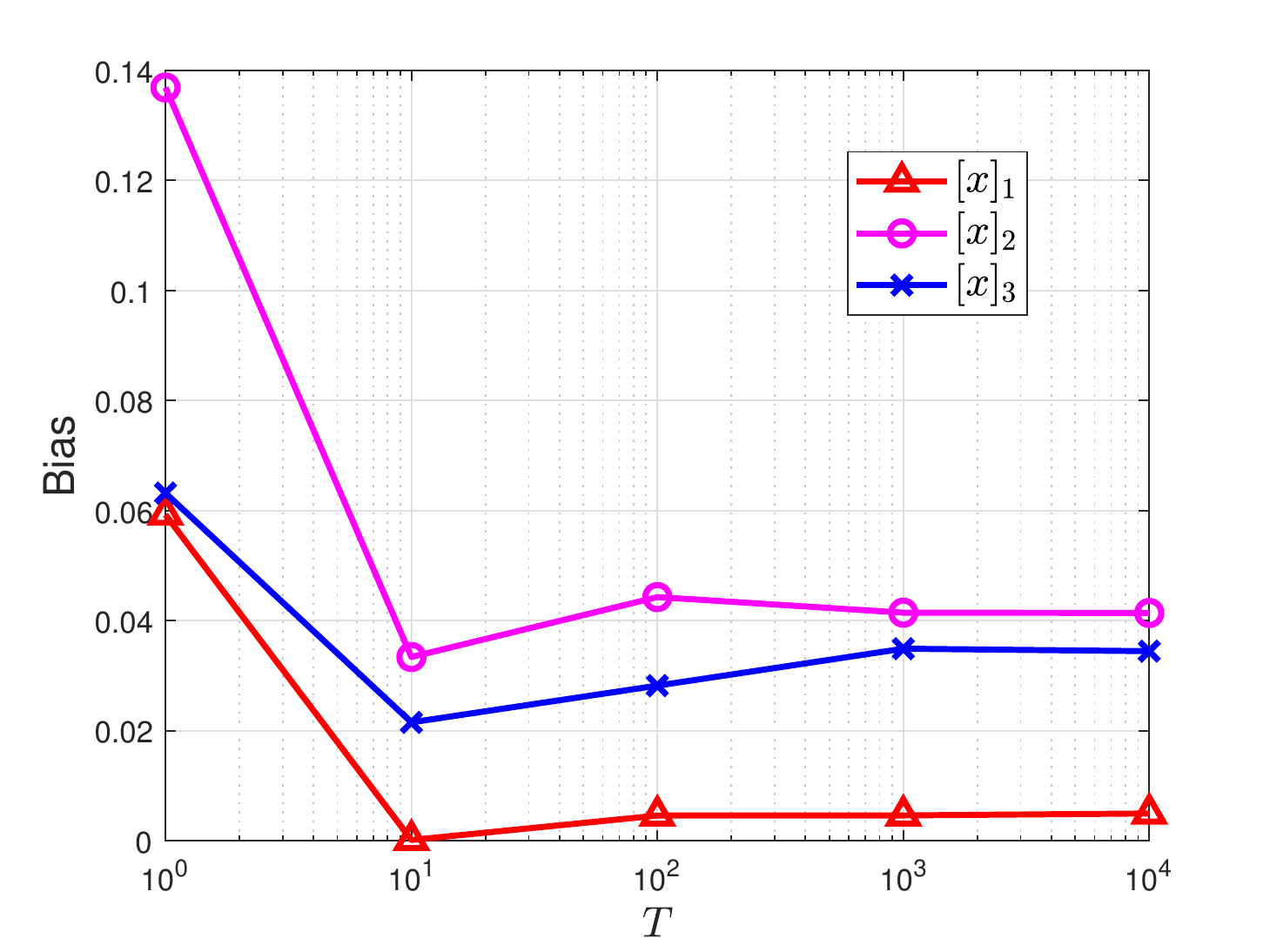}
		\caption{S-LS}
		\label{asymptotic:Bias_S_LS}
	\end{subfigure}
	\caption{Biases of first-step estimators in the large sample case.}
	\label{asymptotic:Bias}
\end{figure*}

\section{Simulations and Discussions} \label{simulations}
In this section, we perform simulations to verify our theoretical developments. Throughout our simulations, we let $10$ sensors sit at 
\begin{align*}
a_1&=[5~0~5]^\top,~a_2=[5~5~-5]^\top, ~a_3=[5~-5~5]^\top, \\
a_4&=[5~0~0]^\top,~a_5=[5~5~5]^\top, ~a_6=[-5~0~-5]^\top, \\
a_7&=[-5~-5~5]^\top,~a_8=[-5~5~-5]^\top, \\
a_9&=[-5~0~0]^\top,~a_{10}=[-5~-5~-5]^\top,
\end{align*}
and the true coordinates of the object is set as $x^o=[6~6~6]^\top$. For the finite sample case, each sensor only observes once to localize the object. While for testing the asymptotic performances, we suppose that each sensor makes totally $T$ rounds of observations, and by letting $T$ tend to infinity, the TOA measurements can be large enough. Note that the sensor deployment scheme here is exactly the case that Example~\ref{example_fix_sensors} illustrates, and thus satisfies Assumption~\ref{convergence_of_sample_distribution}. Further, since the ten sensors are not coplanar, Assumption~\ref{deployment_of_sensor} also holds. We take the Cramer-Rao lower bound (CRLB) as the baseline for the devised estimators' MSEs. Given range measurements $d=[d_1,\ldots,d_m]^\top$, based on the original model~\eqref{range_measurement_model}, the log likelihood function is
\begin{equation*}
\ell (d;x)=m \log \frac{1}{\sqrt{2 \pi}\sigma}-\sum\limits_{i=1}^{m} \frac{\left(d_i-\|a_i-x\| \right)^2}{2\sigma^2},
\end{equation*}
which gives
\begin{equation*}
\frac{\partial \ell (d;x)}{\partial x}=\frac{1}{\sigma^2}\sum\limits_{i=1}^{m} r_i \frac{(x-a_i)^\top}{\|x-a_i\|}.
\end{equation*}
Then we obtain the Fisher information matrix
\begin{equation} \label{Fisher_information_matrix}
\begin{split}
{F}= & \mathbb E\left[\frac{\partial \ell (d;x)}{\partial x} \left( \frac{\partial \ell (d;x)}{\partial x}\right) ^\top \right] \\
= & \frac{1}{\sigma^2} \sum_{i=1}^{m} \frac{(x-a_i)(x-a_i)^\top}{\|x-a_i\|^2},
\end{split}
\end{equation}
and ${\rm CRLB}={\rm tr}(F^{-1})$.

\begin{figure*}[t]
	\centering
	\includegraphics[width=1\textwidth]{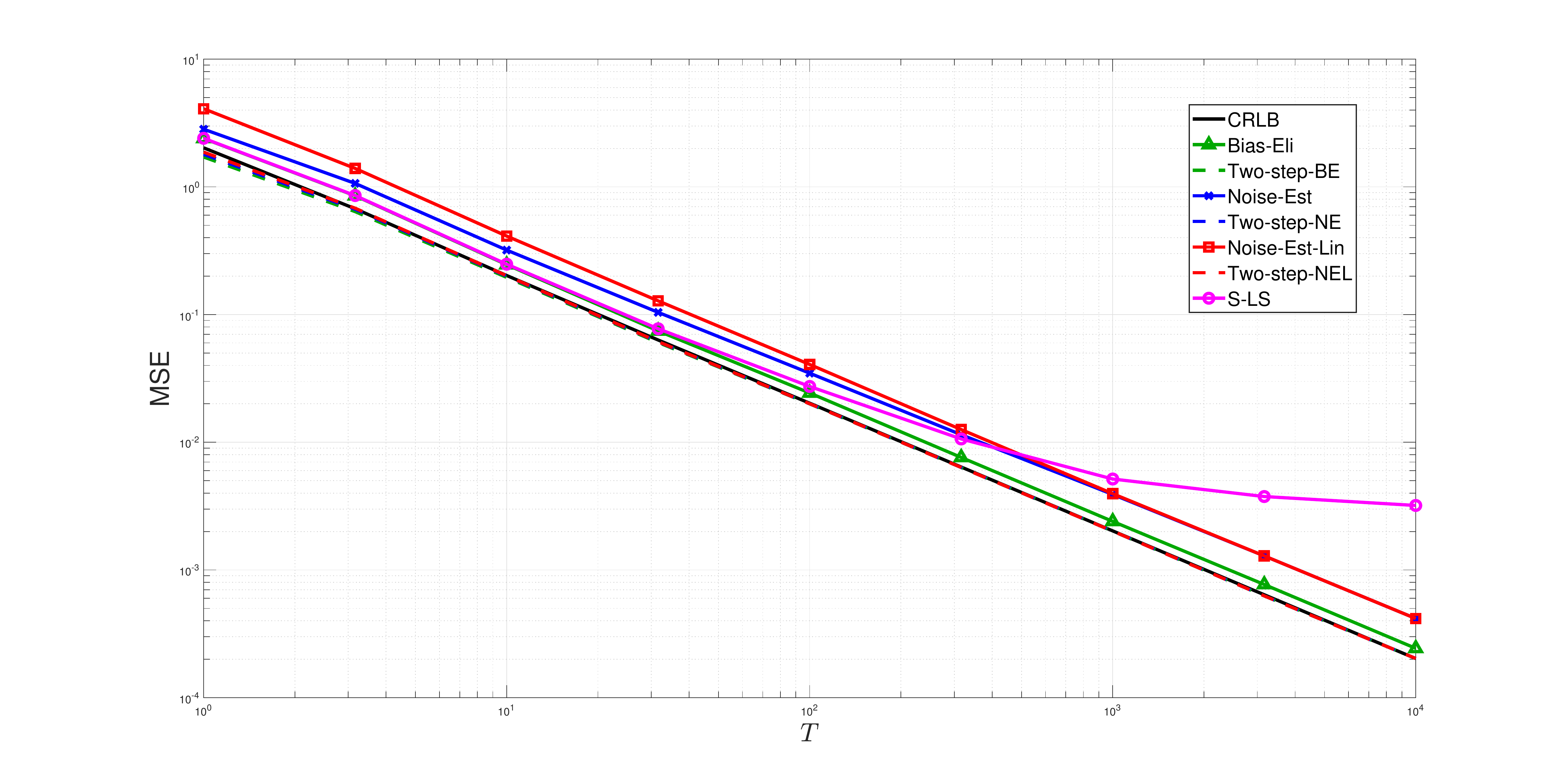}
	\caption{MSE comparison under different numbers of measurements.}
	\label{asymptotic_rmse}
\end{figure*}

\emph{Trial 1.} First, we investigate the biases of the proposed first-step estimators in the finite sample case. The range measurements $d_i,i=1,\ldots,m$ are corrupted by a sequence of i.i.d. Gaussian noises $r_i \sim \mathcal N(0,1)$. We increase the Monte-Carlo runs to show the trends of the average deviation of the estimation from the true position
$\left| \frac{1}{N}\sum_{i=1}^{N}\left( {\hat x(\omega_i) - x^o}\right)\right|$ 
where $N$ is the number of Monte-Carlo runs, and $\hat x(\omega_i)$ is our coordinate estimate in the $i$-th run under a sample $\omega_i$ of the random runs. The results are presented in Fig.~\ref{finite:Bias}: the curve ``$[x]_i$'' represents the average deviation of the $i$-th element of $\hat x$ from $x^o$; the curve ``$\sigma^2$'' represents the average deviation of the noise variance estimation $\hat \sigma^2$ from $\sigma^2$. We see that in the finite sample case, the Bias-Eli and Noise-Est estimators are biased, since with the increase of Monte-Carlo runs, the biases do not go to $0$. While their unconstrained counterparts Bias-Eli-Lin and Noise-Est-Lin are both unbiased for the estimation of $x^o$. The Noise-Est-Lin is biased for the estimation of $\sigma^2$, which coincides with our claim in Subsection~\ref{finite_sample_noise_est}. In the previous section, we show that $\hat x^{\rm BEL}_m=\hat x^{\rm NEL}_m$. This is verified by our simulation results (Fig.~\ref{finite:Bias_Bias_Eli_Lin} and~\ref{finite:Bias_Noise_Est_Lin}). Hence, in the rest of our simulations, we only plot the results of the Noise-Est-Lin estimator and omit that of the Bias-Eli-Lin one.

\emph{Trial 2.} In this trial, we will show the biases of the proposed first-step estimators in the large sample case by letting $T$ go to large. The variance of range measurement noises is $1$. The $T$ is set to be $1$, $10$, $100$, $1000$ and $10000$ respectively, and for each $T$, we run $1000$ Monte-Carlo tests to evaluate the average deviation. Note that the range measurement frequency of UWB systems can reach $1000$ Hz~\cite{grossiwindhager2019snaploc}, which means the time to obtain $10000$ measurements for each sensor is about $10$ seconds. Therefore, this setting is realistic and can be realized when the target is static. We regard the average deviation as the bias, i.e.,
\begin{equation*}
\left|\mathbb E\left[ {\hat x - x^o} \right] \right| 
\approx
\left| \frac{1}{N}\sum_{i=1}^{N}\left( {\hat x(\omega_i) - x^o}\right)\right|,
\end{equation*}
where $N=1000$. The results are shown in Fig.~\ref{asymptotic:Bias}, where Fig.~\ref{asymptotic:Bias_S_LS} plots the biases of the S-LS estimator~\eqref{S-LS}. We see that the biases of the Bias-Eli and Noise-Est estimators tend to $0$ as the increase of $T$, while the S-LS solution~\cite{more1993generalizations,beck2008exact} does not. The difference lies in that the error terms $e_i$ in S-LS problem~\eqref{S-LS} are not mean $0$, while the error terms $\varepsilon_i$ in the Bias-Eli problem~\eqref{Bias_eliminate_problem} Noise-Est problem~\eqref{Noise_Estimate_problem} have $0$ mean. We also note that the Noise-Est and Noise-Est-Lin solutions are both asymptotically unbiased for $\sigma^2$.

\begin{figure*}[!t]
	\centering
	\begin{minipage}[b]{.48\textwidth}
		\centering
		\includegraphics[width=0.96\textwidth]{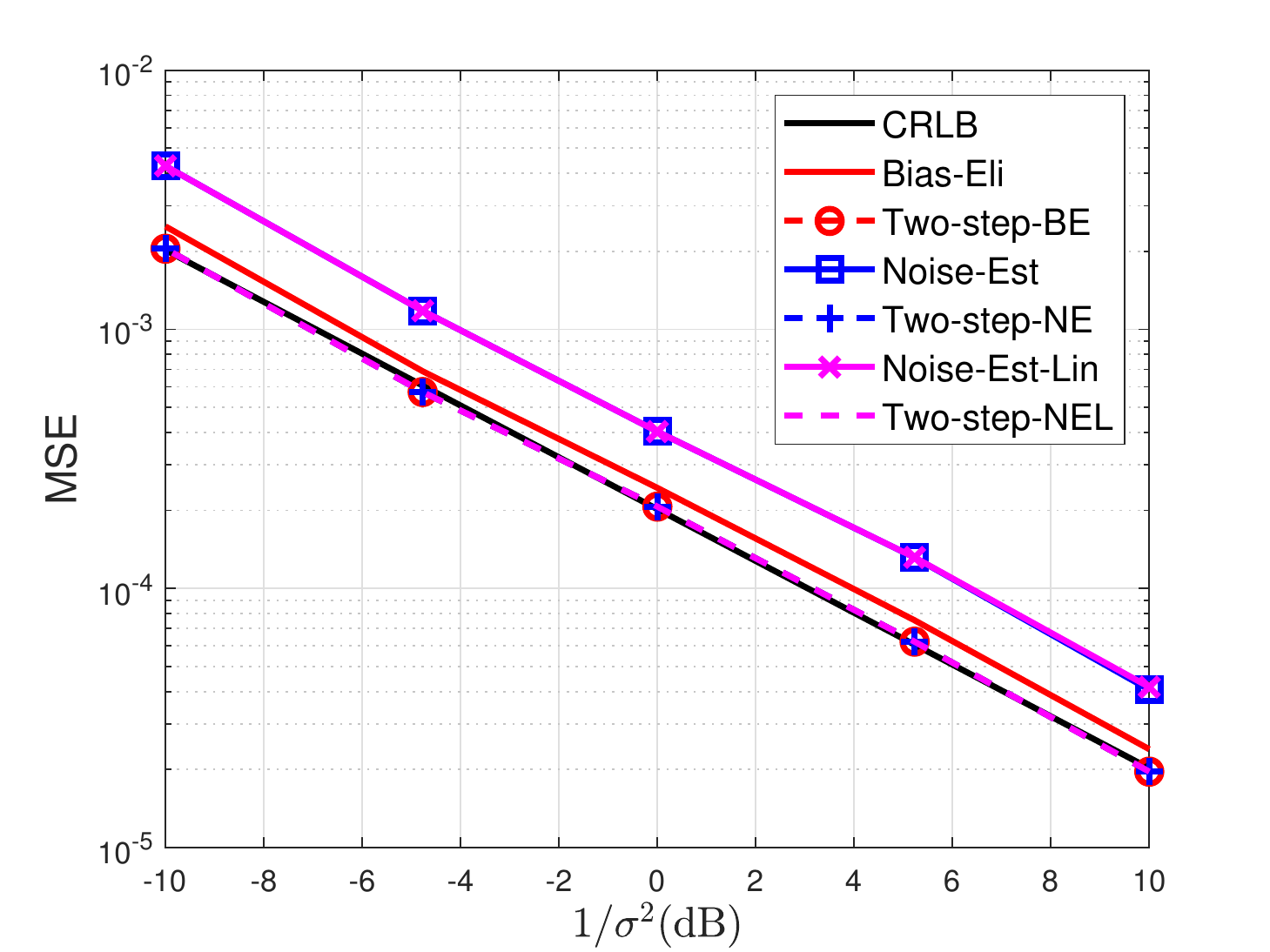}
		\caption{Asymptotic MSE comparison under different noise intensities.}
		\label{asymptotic_rmse_two_step_varied_noise}
	\end{minipage}\qquad
	\begin{minipage}[b]{.48\textwidth}
		\centering
		\includegraphics[width=0.96\textwidth]{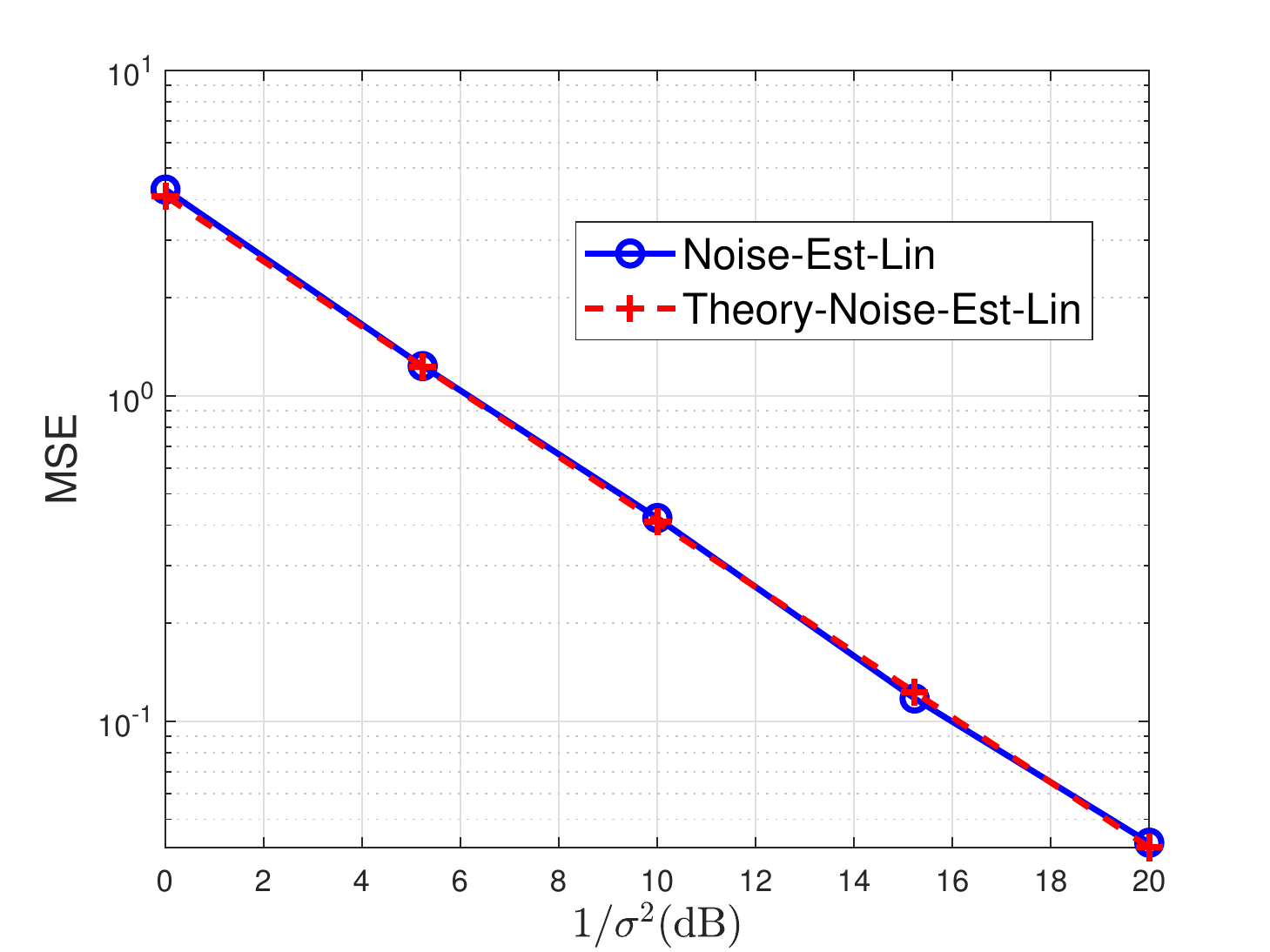}
		\caption{Theoretical and practical MSEs of the Noise-Est-Lin estimator under different noise intensities.}
		\label{theory_rmse_varied_noise}
	\end{minipage}
\end{figure*}
\begin{figure*}[!t]
	\centering
	\begin{minipage}[b]{.48\textwidth}
		\centering
		\includegraphics[width=0.96\textwidth]{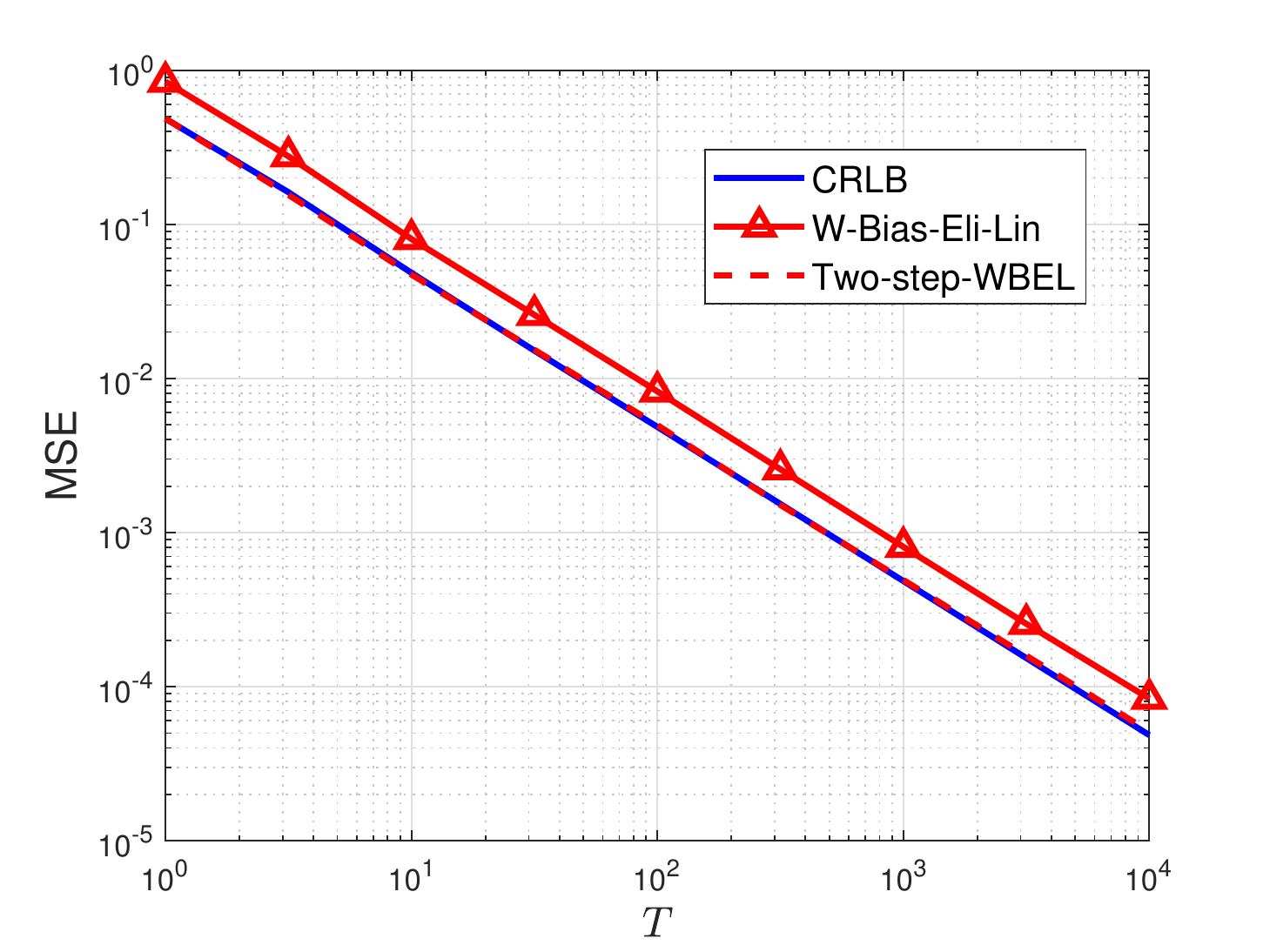}
		\caption{MSE of the W-Bias-Eli-Lin estimator and the corresponding two-step estimator.}
		\label{asymptotic_rmse_WBEL}
	\end{minipage}\qquad
	\begin{minipage}[b]{.48\textwidth}
		\centering
		\includegraphics[width=0.96\textwidth]{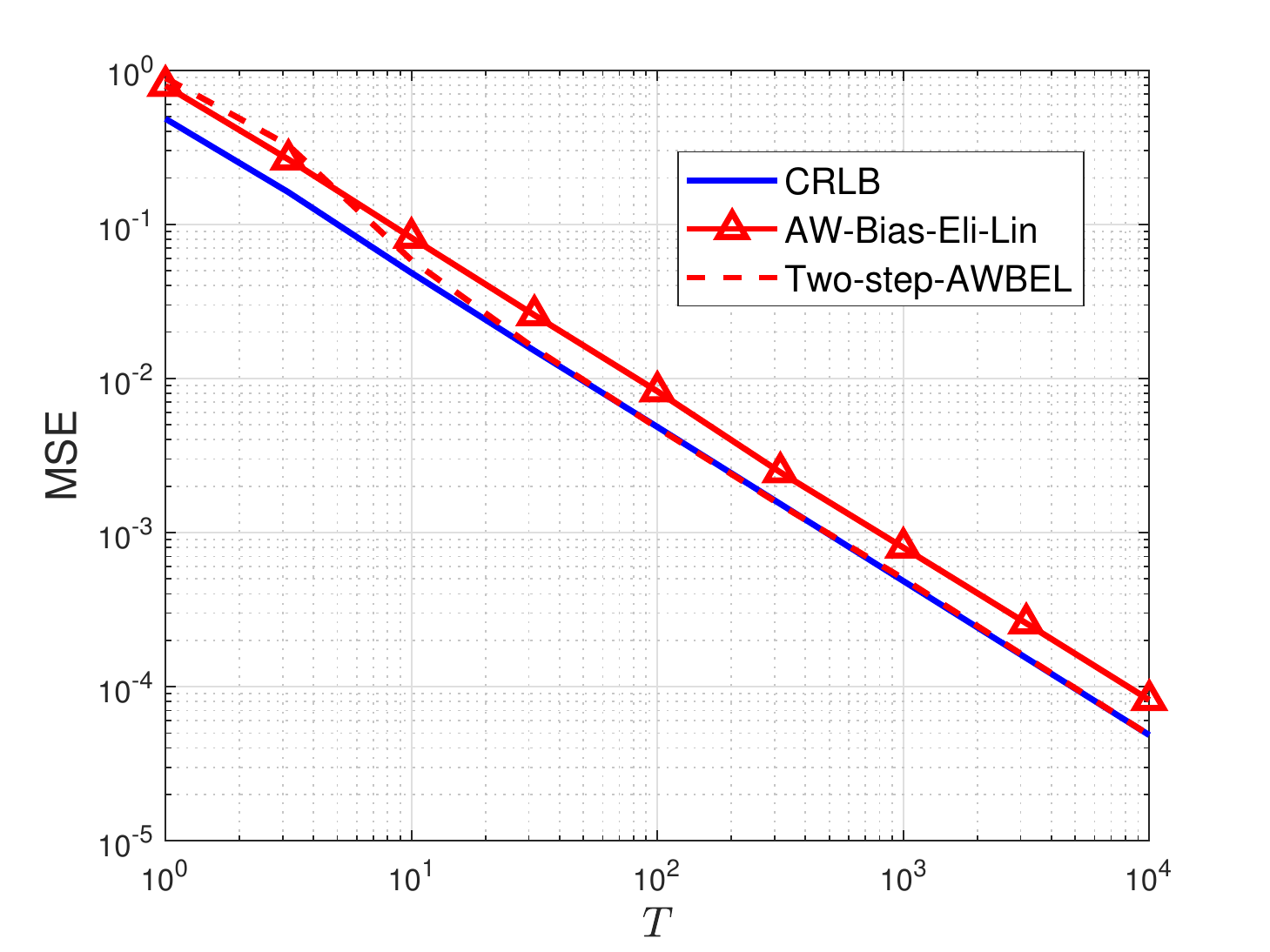}
		\caption{MSE of the AW-Bias-Eli-Lin estimator and the corresponding two-step estimator.}
		\label{asymptotic_rmse_AWBEL}
	\end{minipage}
\end{figure*}

\emph{Trial 3.} Now we are going to present the MSEs of $\hat x$ under varied numbers of measurements. The MSE of $\hat x$ is approximated as follows 
\begin{equation*}
\mathbb E\left[\left\| {\hat x - x^o} \right\|^2\right]
\approx
\frac{1}{N}\sum_{i=1}^{N}{\left\| {\hat x(\omega_i) - x^o} \right\|^2}.
\end{equation*}
We let $\sigma^2=1$ and run $1000$ Monte-Carlo tests to evaluate the MSEs. The results are plotted in Fig.~\ref{asymptotic_rmse}. We note that the proposed first-step estimators are all $\sqrt{m}$-consistent. Nevertheless, the S-LS solution is not. This is because the noise term $e_i$ in the S-LS problem are not mean $0$, which leads to biasedness of the result, even when $T$ goes to infinity. We also note that due to the extra estimation of the noise variance, the MSEs of the Noise-Est and Noise-Est-lin estimators are larger than that of the Bias-Eli estimator. 
An interesting phenomenon is that when $T$ is not large enough, the performance of the Noise-Est estimator is superior to that of the Noise-Est-Lin estimator, which implies that the inequality constraint~\eqref{Noise_Est_constraint2} plays a positive role in these cases. While when $T$ is sufficiently large, the Noise-Est estimate converges to the Noise-Est-Lin one with probability one as shown in~\eqref{limit_probability}. Therefore, they have the same MSEs in the large sample case. We can also see that the proposed two-step estimators perform better than the corresponding first-step estimators. When $T$ is small, the two-step estimators perform even better than the CRLB. This is because in the finite sample case, these estimator are biased and their MSEs are not necessarily larger than the CRLB. Nevertheless, as $T$ increases, all of them become unbiased ones and the MSEs equal the CRLB. It is noteworthy that although the MSEs of the first-step estimators vary obviously, all of the two-step estimators have the same MSEs when $T$ is large.

\emph{Trial 4.} In this trial, we investigate the relations between the improvement induced by the additional one-step Gauss-Newton iteration and the noise variance in the large sample case. The number of observations of each sensor $T$ is set as $10000$, and for each noise variance, we run $1000$ Monte-Carlo tests to approximate the MSEs. The noise variances are chosen as $0.1,0.3,1,3,10$, and the results are presented in Fig.~\ref{asymptotic_rmse_two_step_varied_noise}, where we take the x-axis as $10\log(1/\sigma^2)$. We see from the figure that although the MSEs of the first-step estimators behave differently, all of the two-step estimators have the same MSEs. This is guaranteed by Theorem~\ref{theorem_two_step} which says that any $\sqrt{m}$-consistent estimator can yield an estimate owning the same asymptotic property as the original LS solution by applying a one-step GN iteration. It is also notable that all of the two-step estimators are asymptotically efficient, i.e., their MSEs reach the CRLB.

\emph{Trial 5.} In this trial, we compare the theoretical and practical MSEs of the Noise-Est-Lin estimator under different noise intensities in the finite sample case. We set $T=1$ and run $1000$ Monte-Carlo tests for each choice of noise variance to obtain the practical MSEs. The theoretical MSEs are calculated according to~\eqref{finite_MSE_NE}. We present the results in Fig.~\ref{theory_rmse_varied_noise}. We see that the theoretical and practical MSEs of the Noise-Est-Lin estimator match well. 
%

\emph{Trial 6.} In this trial, we test the performances of the weighted version of the proposed estimators when the measurement variances among sensors are different. The measurement variances of ten sensors are set as $0.01$, $0.04$, $0.09$, $0.16$, $0.25$, $0.36$, $0.49$, $0.64$, $0.81$, and $1$, respectively. For each choice of $T$, we run $1000$ Monte-Carlo tests to calculate the MSEs. When the measurement variances of each sensor are known, a weighted two-step estimation scheme can be utilized to obtain asymptotically efficient solutions. The result is plotted in Fig.~\ref{asymptotic_rmse_WBEL}, from which we see that the first-step estimator is consistent, and the two-step estimator is asymptotically efficient. We note that when $T$ is small, the two-step estimator can also reach the CRLB in this setting, although it is not guaranteed theoretically. When the measurement variances are unknown, we first estimate the variance of each sensor, and then use an approximate weighted two-step method. Specifically, we set $\hat \sigma^2_i=1$ in the case of $T=1$ to avoid the denominator in~\eqref{AWLS} being $0$. The result is presented in Fig.~\ref{asymptotic_rmse_AWBEL}. We see that when $T$ is small, the estimates of noise variances are not precise, and thus the proposed estimator has relatively large MSEs. However, with the increase of $T$, the estimated variances converge to the true values, and the two-step solution can asymptotically reach the CRLB.

\section{Conclusions} \label{conclusion}
In this paper, we investigated the range-based localization problem. First, we proved that under some conditions on the measurement noises and sensor deployment, the LS estimator is strongly consistent and asymptotically efficient. This guarantees that as the measurements increase, the LS estimate can converge to the true object's position with the minimum variance. However, the LS problem is hard to solve. We then devised realizable estimators that achieve the same asymptotic properties as the LS one. The proposed estimators consist of two steps where the first step is to construct a $\sqrt{m}$-consistent estimator, and the second step is applying a one-step GN iteration. For the first step, we proposed two $\sqrt{m}$-consistent estimators (Bias-Eli and Noise-Est) which involve solving two QCQP optimization problems. Moreover, we showed that the closed-form estimates (Bias-Eli-Lin and Noise-Est-Lin) obtained by discarding the constraints of the QCQP problems are also $\sqrt{m}$-consistent. We noticed that when $m$ is fixed, the nonlinear least squares solutions Bias-Eli and Noise-Est are biased for the estimation of $x^o$, while the linear least squares solutions Bias-Eli-Lin and Noise-Est-Lin are unbiased. As the number of sensors goes to infinity, all of the proposed four first-step estimators are asymptotically unbiased, based on which the second-step estimators can achieve the CRLB. 

\appendices

\section{Existence of some tail products and tail norms} \label{existence_of_tail_product_norm}
Given Assumption~\ref{convergence_of_sample_distribution}, the Helly-Bray Theorem~\cite{billingsley2013convergence} shows the convergence of the sample mean of any bounded, continuous and real-valued function.
\begin{lemma}[Helly-Bray Theorem~\cite{billingsley2013convergence}] \label{Helly_Bray_Thoerem}
	Let $P_1,P_2,\ldots$ be probability measures on a sample space $\Omega$. Then $P_m$ converges weakly to $P$ if and only if
	\begin{equation*}
	\int_{\Omega} gd P_m \rightarrow \int_{\Omega} gd P 
	\end{equation*}
	for all bounded, continuous and real-valued functions on $\Omega$.
\end{lemma}
The existences of the following tail products and tail norms are based on Assumptions~\ref{compact_set_interior},~\ref{convergence_of_sample_distribution}, and Lemma~\ref{Helly_Bray_Thoerem}.

\noindent $\left\| f(x)-f(x^o) \right\|_t$: By Lemma~\ref{Helly_Bray_Thoerem}, $m^{-1} \sum_{i=1}^{m} f_i(x_1) f_i(x_2)$ converges to $\mathbb E_{a \sim \mu} \left[ \|a-x_1\| \|a-x_2\|\right]$, i.e., $\left\langle f(x_1),f(x_2)\right\rangle_t $ for all $x_1$ and $x_2$ in $\mathcal X$, where $\mathbb E_{a \sim \mu}$ is taken over $a$. Hence, the tail norm $\left\| f(x)-f(x^o) \right\|_t$ is well defined.

\noindent $\left\|  f^2(x)-f^2(x^o)\right\|_t $: 
By Lemma~\ref{Helly_Bray_Thoerem}, $\left\langle f^2(x_1),f^2(x_2)\right\rangle_t$ exists for all $x_1,x_2 \in \mathcal X$. Then the tail norm $\left\|  f^2(x)-f^2(x^o)\right\|_t$ exists.

\noindent $\left\langle f^3(x_1),f^3(x_2)\right\rangle_t$ and $\left\langle f^2f(x_1,x_2),f^2f(x_1,x_2)\right\rangle_t$: This is straightforward from Lemma~\ref{Helly_Bray_Thoerem}.

\noindent $\left\langle p,q \right\rangle_t$, where $p \in \left\lbrace f(x_1),f_j'(x_1),f_{jk}''(x_1)\right\rbrace $ and $q \in \left\lbrace f(x_2),f_j'(x_2),f_{jk}''(x_2)\right\rbrace $ for all $x_1$ and $x_2$ in $\mathcal X$: Based on Assumption~\ref{compact_set_interior}, for any $x \in \mathcal X$, $\|a-x\|$, $\frac{x-a}{\|a-x\|}$, and $\frac{I_n}{\|a-x\|}-\frac{(x-a)(x-a)^\top}{\|a-x\|^3}$ are bounded and continuous in $\mathcal A$. Hence these tail products exist by Lemma~\ref{Helly_Bray_Thoerem}. 

\section{Proof of Theorem~\ref{theorem_LS}} \label{proof_theorem_ls}
Some conditions, see (a)-(d) below, are given in~\cite{jennrich1969asymptotic} to guarantee that a general nonlinear LS estimator is consistent and asymptotically normal. Here we rephrase them using the notations in this paper. 
\begin{enumerate}
	\item [(a)] The functions $f_i(x)$ are continuous on a compact subset of a Euclidean space and the noises $r_i$ are i.i.d. with zero mean and finite variance. 
	\item [(b)] The tail product of $f=(f_i)$ with itself exists and that $\|f(x)-f(x^o)\|^2_t$ has a unique minimum at $x=x^o$.
	\item [(c)] The derivatives $f'_{ji}$ and $f''_{jki}$ exist and continuous on $\mathcal X$ and that all tail products of the form $\left\langle p,q \right\rangle_t$, where $p \in \left\lbrace f(x_1),f_j'(x_1),f_{jk}''(x_1)\right\rbrace $, $q \in \left\lbrace f(x_2),f_j'(x_2),f_{jk}''(x_2)\right\rbrace $ exist for all $x_1$ and $x_2$ in $\mathcal X$.
	\item [(d)] The true parameter vector $x^o$ is an interior point of $\mathcal X$ and the matrix $M(x^o)$ is nonsingular. 
\end{enumerate}

Note that the Assumptions~\ref{assumption:Gaussian_noise}-\ref{nonsingular_assumption} in this paper are specific in the context of TOA-based localization. All what we need to do is to show that under Assumptions~\ref{assumption:Gaussian_noise}-\ref{nonsingular_assumption}, the general conditions (a)-(d) in~\cite{jennrich1969asymptotic} can be satisfied.
The existence of $\left\langle p,q \right\rangle_t$, where $p \in \left\lbrace f(x_1),f_j'(x_1),f_{jk}''(x_1)\right\rbrace $, $q \in \left\lbrace f(x_2),f_j'(x_2),f_{jk}''(x_2)\right\rbrace $ for all $x_1$ and $x_2$ in $\mathcal X$ is illustrated in Appendix~\ref{existence_of_tail_product_norm}. Then, it is straightforward that, with Assumptions~\ref{assumption:Gaussian_noise}-\ref{unique_solution_assumption}, the conditions (a)-(c) in~\cite{jennrich1969asymptotic} hold. The only condition that needs to be further verified is Assumption (d), i.e., $M(x^o)$ is nonsingular. Here, we give the proof of 2D case. The argument of 3D case is similar and will be omitted. For any $x \neq 0$, define $\mathcal A^{\perp}_x=\left\lbrace a \in \mathcal A \mid \langle x,a-x^o \rangle = 0 \right\rbrace$, and $\overbar {{\mathcal A}^{\perp}_x}=\mathcal A \setminus \mathcal A^{\perp}_x$. Given Assumption~\ref{nonsingular_assumption}, we have $\mu\left(\overbar {{\mathcal A}^{\perp}_x} \right) >0$. Then we can decompose $M(x^o)$ as
\begin{align*}
M(x^o) = &\lim\limits_{m \rightarrow \infty} \frac{1}{m} \sum_{i=1}^{m} \nabla f_i(x^o) \nabla f_i(x^o)^\top \\
 =&\mathbb E_{a \sim \mu}\left[\nabla \|a-x^o\| \nabla \|a-x^o\|^\top \right] \\
 = &\underbrace{\int_{\mathcal A^{\perp}_x} \nabla \|a-x^o\| \nabla \|a-x^o\|^\top d \mu(a)}_{:= M^{\perp}_x(x^o)} \\
& +\underbrace{\int_{\overbar {{\mathcal A}^{\perp}_x}} \nabla \|a-x^o\| \nabla \|a-x^o\|^\top d \mu(a)}_{:=\overbar {{M}^{\perp}_x}(x^o)}. 
\end{align*}
Since $\mu\left(\overbar {{\mathcal A}^{\perp}_x} \right) >0$, we have $\overbar {{M}^{\perp}_x}(x^o) \neq 0$. Therefore,
\begin{align*}
x^\top M(x^o) x &= x^\top \left(M^{\perp}_x(x^o)+ \overbar {{M}^{\perp}_x}(x^o)\right)  x \\
& = x^\top \overbar {{M}^{\perp}_x}(x^o) x \\
& >0,
\end{align*}
which implies $M(x^o)$ is positive-definite and completes the proof. 

\section{Proof of Theorem~\ref{consistency_of_bias_eli}}
\label{proof_of_BE_consistency}
First, we show that ${\hat x}^{\rm BE}_m$ is consistent. Based on Assumption~\ref{convergence_of_sample_distribution}, the function $P(x):=\left\|  f^2(x)-f^2(x^o)\right\|_t  ^2$ exists. See detailed arguments in Appendix~\ref{existence_of_tail_product_norm}. Moreover, since $P(x)$ is a limit of uniformly convergent sequence of continuous functions, $P(x)$ is continuous. Recalling that $\left\| f(x)-f(x^o) \right\|^2_t$ has a unique minimum at $x=x^o$, it can be verified that $x=x^o$ is also the unique minimum of $P(x)$. Denote the objective function in~\eqref{Bias_eliminate_problem} as $P_m(x)$. Let $f^k(x):=\left(f^k_i(x)\right) $, where $f^k_i(x)=\|a_i-x\|^k$.
Before giving the convergence of $P_m(x) $, we will show that $\left\langle f^2(x^o)-f^2(x),\varepsilon \right\rangle_t  =0$ for almost every $\varepsilon:=(\varepsilon_i)$, which depends on the following lemma.
\begin{lemma}[{\cite[Theorem 3]{jennrich1969asymptotic}}]  \label{tail_product_with_noise}
	If a sequence of noises $e$ satisfies Assumption~\ref{assumption:Gaussian_noise}, and if the tail norm of a sequence $f$ of real numbers exists, then $\langle f,e \rangle_t$ exists and equals $0$ for almost every $e$.
\end{lemma}
Define $f^2f(x_1,x_2):=\left(f_i^2(x_1)f_i(x_2) \right) $. Given Assumption~\ref{convergence_of_sample_distribution}, $\left\langle f^3(x^o),f^3(x^o)\right\rangle_t$ and $\left\langle f^2f(x,x^o),f^2f(x,x^o)\right\rangle_t$ exist for all $x \in \mathcal X$. See detailed arguments in Appendix~\ref{existence_of_tail_product_norm}. Then in virtue of Lemma~\ref{tail_product_with_noise}, we have
\begin{align*}
&m^{-1} \sum_{i=1}^{m} \left(f^2_i(x^o)-f^2_i(x) \right) \varepsilon_i \\
=&m^{-1} \sum_{i=1}^{m} \left(f^2_i(x^o)-f^2_i(x) \right) \left(2f_i(x^o)r_i+r_i^2-\sigma^2 \right)  \\
\rightarrow & 2 \left\langle f^3(x^o),r \right\rangle_t -2 \left\langle f^2f(x,x^o),r \right\rangle_t \\
&+ \left\langle f^2(x^o)-f^2(x),r^2-\sigma^2 \right\rangle_t\\
=&0,
\end{align*}
where $r^k:=\left( r_i^k\right) $.
Hence, for almost every $\varepsilon$, we have
\begin{align*}
m^{-1} P_m(x) &=m^{-1}\sum_{i=1}^{m}\left( f_i(x)^2-f_i(x^o)^2-\varepsilon_i \right) ^2 \\
&\rightarrow \left\|f^2(x^o)-f^2(x) \right\|_t^2 +\|\varepsilon\|_t^2 \\
& =P(x)+\|\varepsilon\|_t^2,
\end{align*}
uniformly for $x \in \mathcal X$, where 
\begin{align*}
\|\varepsilon\|_t^2  =&\lim\limits_{m \rightarrow \infty} m^{-1} \sum_{i=1}^{m} \varepsilon_i^2  \\
=&\lim\limits_{m \rightarrow \infty} m^{-1} \sum_{i=1}^{m} \left(2f_i(x^o)r_i +r_i^2-\sigma^2\right) ^2  \\
 =&4\left\langle f^2(x^o),r^2-\sigma^2 \right\rangle_t+4\sigma^2\|f(x^o)\|^2_t+2\sigma^4 \\
&+4\left\langle f(x^o),r^3 \right\rangle_t -4\sigma^2\left\langle f(x^o),r \right\rangle_t\\
 =&4 \sigma^2 \left\|f(x^o) \right\|_t^2  +2\sigma^4. 
\end{align*}
Note that $\hat x^{\rm BE}_m$ minimizes $m^{-1} P_m(x)$ for any $m \in \mathbb N$. Then $\left(\hat x^{\rm BE}_m \right)$ forms a sequence of minimizers of $m^{-1} P_m(x)$. Let $x'$ be a limit point of the sequence $\left(\hat x^{\rm BE}_m \right) $, and let $\left( \hat x^{\rm BE}_{m_k}\right)$ be any subsequence which converges to $x'$. By the continuity of $P$ and the uniform convergence of $m^{-1}P_m$ to $P+\|\varepsilon\|_t^2$, $m_k^{-1}P_{m_k}\left(\hat x^{\rm BE}_{m_k} \right)  \rightarrow P(x')+\|\varepsilon\|_t^2$ as $k \rightarrow \infty$. 
Since $\hat x^{\rm BE}_{m_k}$ is the global minimizer of $m_k^{-1} P_{m_k}(x)$, $m_k^{-1} P_{m_k}\left( \hat x^{\rm BE}_{m_k}\right)  \leq m_k^{-1} P_{m_k}(x^o)$. It follows that by letting $k \rightarrow \infty$, $P(x')+\|\varepsilon\|_t^2 \leq P(x^o)+\|\varepsilon\|_t^2=\|\varepsilon\|_t^2$. Hence $P(x')=0$. Since $P$ has a unique minimum at $x^o$, $x'=x^o$. Thus for almost every $\varepsilon$, $\hat x^{\rm BE}_{m} \rightarrow x^o$. 

Since ${\hat x}^{\rm BE}_m$ is consistent, and $\nabla^2 P_m(x)$ is continuous, as a basic conclusion of large sample analysis in nonlinear estimation, when $m \rightarrow \infty$, we have: 
\begin{equation}  \label{convergencen_of_mse}
\mathbb {MSE}\left( \hat x^{\rm BE}_m\right)= \overline F''^{-1}\mathbb E \left[\nabla P_m(x^o) \nabla P_m(x^o)^\top \right] \overline F''^{-1}, 
\end{equation}
where $\overline F''=\lim_{m \rightarrow \infty} \nabla^2 P_m(x^o)$. For a detailed treatment, one can refer to~\cite{viberg1991sensor}. We are now on the point to calculate $\mathbb E \left[\nabla P_m(x^o) \nabla P_m(x^o)^\top \right]$ and $\overline F''$. The gradient of $P_m(x)$ (recall that $P_m(x)$ is the objective function in~\eqref{Bias_eliminate_problem}) is as follows:
\begin{equation*}
\nabla P_m(x)=2 \left[
\begin{array}{cc}
I & 2x
\end{array}
\right] A^\top (Ay-b),  
\end{equation*}
\hspace{-2mm}which gives 
\begin{equation*}
\mathbb E \left[\nabla P_m(x^o) \nabla P_m(x^o)^\top \right]=4 \left[
\begin{array}{cc}
I & 2x^o
\end{array}
\right] A^\top \Lambda A \left[
\begin{array}{cc}
I \\
2{x^o}^\top
\end{array}
\right],
\end{equation*}
where 
\begin{equation} \label{covariance_of_noises}
\Lambda={\rm diag} \left(4f_1^2(x^o)\sigma^2+2 \sigma^4,\ldots,4f_m^2(x^o)\sigma^2+2 \sigma^4 \right).
\end{equation}
Then we calculate the second gradient as
{\small
\begin{equation*}
\nabla^2 P_m(x)=2 \left[
\begin{array}{cc}
I & 2x
\end{array}
\right] A^\top A \left[
\begin{array}{cc}
I \\
2{x}^\top
\end{array}
\right] +4\left[A^\top (Ay-b) \right] _{n+1} I,
\end{equation*}}
which gives 
\begin{equation*}
\overline F''=2 \left[
\begin{array}{cc}
I & 2x^o
\end{array}
\right] A^\top A \left[
\begin{array}{cc}
I \\
2{x^o}^\top
\end{array}
\right].
\end{equation*} 
Note that (in virtue of Assumption~\ref{convergence_of_sample_distribution})
\begin{equation} \label{convergence_of_AA}
\lim\limits_{m \rightarrow \infty} \frac{A^\top A}{m}=\mathbb E_{a \sim \mu} \left[\left[\begin{array}{cc} 
-2a \\
1
\end{array}
\right] \left[ \begin{array}{cc} 
-2a^\top & 1
\end{array}\right] \right] 
\end{equation}
and 
{\small
\begin{equation} \label{convergence_of_APA}
\begin{split}
&\lim\limits_{m \rightarrow \infty} \frac{A^\top \Lambda A}{m} \\
&=\mathbb E_{a \sim \mu} \left[ \left(4\|a-x^o\|^2 \sigma^2+2 \sigma^4 \right) \left[\begin{array}{cc} 
-2a \\
1
\end{array}
\right] \left[ \begin{array}{cc} 
-2a^\top & 1
\end{array}\right] \right] ,
\end{split}
\end{equation}} 
\hspace{-2mm}where $\mathbb E_{a \sim \mu}$ is taken over $a$ with respect to $\mu$. Hence, we have $m {\rm tr}\left(\mathbb {MSE}\left( \hat x^{\rm BE}_m\right) \right)$ converges to a constant as $m$ goes to infinity, which implies that $m {\rm tr}\left(\mathbb {MSE}\left( \hat x^{\rm BE}_m\right) \right)$ is bounded, i.e., there exists a $\varphi >0$ such that $m {\rm tr}\left(\mathbb {MSE}\left( \hat x^{\rm BE}_m\right) \right) < \varphi$ for arbitrary $m$.
In addition, according to Chebyshev's inequality, we have 
\begin{equation}
P\left(\sqrt{m} \|{\hat x}^{\rm BE}_m-x^o\| >N\right) \leq \frac{m {\rm tr}\left(\mathbb {MSE}\left( \hat x^{\rm BE}_m\right) \right) }{N^2}, 
\end{equation}
for any $N>0$. Therefore, for any $\epsilon>0$, we can select $N$ as $\sqrt{\frac{\varphi}{\epsilon}}$ to ensure
\begin{equation*}
P\left(\sqrt{m} \|{\hat x}^{\rm BE}_m-x^o\| >N\right) < \frac{\varphi }{N^2} = \epsilon
\end{equation*}
for arbitrary $m$, which shows the $\sqrt{m}$-consistency of $\hat x^{\rm BE}_m$ and completes the proof. 

\section{Proof of Theorem~\ref{consistency_of_linear_bias_eli_estimate}}
\label{proof_of_BEL_consistency}
The proof is supported by the following lemma. 
\begin{lemma} \label{property_of_bounded_variance}
	Let $\{X_k\}$ be a sequence of independent random variables with $\mathbb E[X_k]=0$ and $\mathbb E\left[{X_k}^2 \right]  \leq \varphi <\infty$ for all $k$. Then, there holds $\sum_{k=1}^{m}X_k/\sqrt{m}=O_p(1)$.
\end{lemma}
\begin{proof}
	Let $\bar X_m=\sum_{k=1}^{m}X_k/\sqrt{m}$, then we have $\mathbb E\left[{{{}\bar X}_k}^2 \right]  \leq \varphi$. Hence, by the Chebyshev's inequality, we obtain
	\begin{equation*}
	P\left(|\bar X_m|\geq N \right) \leq \frac{\varphi}{N^2}, 
	\end{equation*}
	for any $N>0$. For any $\epsilon >0$, we can set $N=\sqrt{\frac{\varphi}{\epsilon}}$ such that $P\left(|\bar X_m|\geq N \right) \leq \epsilon$ for all $m$, which completes the proof. 
\end{proof}
Recall that {\small $\hat y^{\rm BEL}_m=(A^\top A)^{-1} A^\top b=\left(\frac{1}{m} A^\top A\right)^{-1}\left(\frac{1}{m} A^\top b\right)  $}. For $\frac{1}{m} A^\top b$ we have 
\begin{align*}
\frac{1}{m}A^\top b =&\frac{1}{m}A^\top \left[ 
\begin{array}{cc}
-2a_1^\top x^o+\|x^o\|^2  \\
\vdots \\
-2a_m^\top x^o+\|x^o\|^2
\end{array}
\right] \\
&+\frac{1}{m}A^\top \left[ 
\begin{array}{cc}
2\|a_1-x^o\|r_1+r_1^2-\sigma^2  \\
\vdots \\
2\|a_m-x^o\|r_m+r_m^2-\sigma^2
\end{array}
\right]\\
=& \frac{1}{m}A^\top A \left[ 
\begin{array}{cc}
x^o  \\
\|x^o\|^2
\end{array}
\right]+\frac{1}{\sqrt{m}} O_p\left(1 \right),
\end{align*}
where the second equality is based on Lemma~\ref{property_of_bounded_variance}. Thus, we obtain
\begin{equation*}
\sqrt{m}\left( \hat y^{\rm BEL}_m-\left[ 
\begin{array}{cc}
x^o  \\
\|x^o\|^2
\end{array}
\right]\right)= \left( \frac{1}{m}A^\top A\right)^{-1}O_p\left(1 \right)=O_p\left(1 \right)
\end{equation*}
where the second equality holds because $ \frac{1}{m}A^\top A$ converges to a constant matrix, see~\eqref{convergence_of_AA}. Hence, $\hat y^{\rm BEL}_m$ converges to $[{x^o}^\top ~\|x^o\|^2]^\top $ at a rate of $1/\sqrt{m}$, which completes the proof.

\section{Proof of Theorem~\ref{consistency_of_Noise_Est}}
\label{proof_of_NE_consistency}
First, we show that ${\hat y}^{\rm NE}_m$ is consistent. Denote the objective function in~\eqref{Noise_Estimate_problem} as $Q_m(x,c)$. 
In virtue of Lemma~\ref{tail_product_with_noise}, for almost every $\varepsilon$, we have
{\small\begin{align*}
	\frac{1}{m} Q_m(x,c) &=\frac{1}{m}\sum_{i=1}^{m}\left( f_i(x)^2-f_i(x^o)^2-\varepsilon_i -\sigma^2 +c \right) ^2 \\
	&\rightarrow \underbrace{\mathbb E_{a \sim \mu} \left[(\|a-x\|^2-\|a-x^o\|^2 -\sigma^2+c)^2 \right]}_{:=Q(x,c)}+\|\varepsilon\|_t^2 ,
	\end{align*}}
\hspace{-2mm}uniformly for $x \in \mathcal X$ and $c \in \mathbb R$.

Note that $[{{}\hat x^{\rm NE}_m}^\top ~\hat \sigma^{\rm NE}_m ]^\top $ minimizes $m^{-1} Q_m(x,c)$ for any $m \in \mathbb N$. Then $\left([{{}\hat x^{\rm NE}_m}^\top ~\hat \sigma^{\rm NE}_m ]^\top \right)$ forms a sequence of minimizers of $m^{-1} Q_m(x,c)$. Let $[x'^\top ~c']^\top$ be a limit point of the sequence $\left([{{}\hat x^{\rm NE}_m}^\top ~\hat \sigma^{\rm NE}_m ]^\top \right) $, and let $\left( [{{}\hat x^{\rm NE}_{m_k}}^\top ~\hat c^{\rm NE}_{m_k} ]^\top\right)$ be any subsequence which converges to $[x'^\top ~c']^\top$. 
Then, $m_k^{-1}Q_{m_k}\left(\hat x^{\rm NE}_{m_k},\hat c^{\rm NE}_{m_k} \right)  \rightarrow Q(x',c')+\|\varepsilon\|_t^2$ as $k \rightarrow \infty$. 
Since $[{{}\hat x^{\rm NE}_m}^\top ~\hat \sigma^{\rm NE}_m ]^\top $ is a minimizer of $m^{-1} Q_m(x,c)$, $m_k^{-1} Q_{m_k}\left( \hat x^{\rm BE}_{m_k},\hat c^{\rm BE}_{m_k}\right)  \leq m_k^{-1} Q_{m_k}(x^o,\sigma^2)$. It follows that by letting $k \rightarrow \infty$, $Q(x',c')+\|\varepsilon\|_t^2 \leq Q(x^o,\sigma^2)+\|\varepsilon\|_t^2=\|\varepsilon\|_t^2$. Therefore, we obtain $Q(x',c')=0$. Based on the following lemma, we have $[x'^\top ~c']^\top=[{x^o}^\top ~\sigma^2]^\top$.
\begin{lemma}
	Given Assumption~\ref{deployment_of_sensor}, the function $Q(x,c)$ has a unique minimum at $[{x^o}^\top ~\sigma^2]^\top$. 
\end{lemma}
\begin{proof}
	For any $x \in \mathcal X$ and $c \in \mathbb R$, define $\mathcal A_{x,c}=\left\lbrace a\in \mathcal A \mid \|a-x\|^2-\|a-x^o\|^2 -\sigma^2+c=0 \right\rbrace $. It can be verified that $\mathcal A_{x,c}$ is either an empty set or a line (plane) in 2D (3D) case. Suppose these exists a $[x'^\top ~c']^\top \neq [{x^o}^\top ~\sigma^2]^\top$ such that $Q(x',c')=0$. If $c'=\sigma^2$, we have $x'=x^o$ since $\mathbb E_{a \sim \mu} \left[(\|a-x\|^2-\|a-x^o\|^2)^2 \right]$ has a unique minimum at $x=x^o$, which leads to a contradiction. If $c' \neq \sigma^2$, we have $x' \neq x^o$. Note that $Q(x',c')=0$. Hence, $\mu\left( \mathcal A_{x',c'} \right) =1$, which contradicts Assumption~\ref{deployment_of_sensor} and completes the proof.
\end{proof}
As a result, for almost every $\varepsilon$, $[{{}\hat x^{\rm NE}_m}^\top ~\hat \sigma^{\rm NE}_m ]^\top \rightarrow [{x^o}^\top ~\sigma^2]^\top$, i.e., $\hat y^{\rm NE}_m \xrightarrow{\text{a.s.}} \bar y^o$. The Lagrangian function of~\eqref{Noise_Estimate_problem2} is 
\begin{equation} \label{Lagrangian_function}
L(\bar y,\lambda)=\|A \bar y-\bar b\|^2 +\lambda \left( \bar y^\top D \bar y+2g^\top \bar y\right) .
\end{equation}
Then $\hat y^{\rm NE}_m$ needs to meet the stationary condition:
\begin{equation} \label{stationary_condition}
2 A^\top A \hat y^{\rm NE}_m-2 A^\top \bar b+2 \lambda D \hat y^{\rm NE}_m +2 \lambda g=0. 
\end{equation}
Since $\hat y^{\rm NE}_m \xrightarrow{\text{a.s.}} \bar y^o$, $\lim\limits_{m \rightarrow \infty} {{}\hat y^{\rm NE}_m}^\top D \hat y^{\rm NE}_m+2g^\top \hat y^{\rm NE}_m<0$ w.p.1. According to the complementary slackness condition, we have that $\lim\limits_{m \rightarrow \infty} P(\lambda^*=0)=1$, where $\lambda^*$ is the optimal Lagrange multiplier. Based on~\eqref{stationary_condition} and $\lim\limits_{m \rightarrow \infty} P(\lambda^*=0)=1$, we have 
\begin{equation} \label{limit_probability}
\lim\limits_{m \rightarrow \infty} P\left( \hat y^{\rm NE}_m=\hat y^{\rm NE0}_m\right)=1 , 
\end{equation}
where $\hat y^{\rm NE0}_m:=(A^\top A )^{-1} A^\top \bar b$.
Now we show that $\hat y^{\rm NE0}_m$ is a $\sqrt{m}$-consistent estimate of $\bar y^o$. The proof is similar to that of Theorem~\ref{consistency_of_linear_bias_eli_estimate}. Note that $\hat y^{\rm NE0}_m=\left(\frac{1}{m} A^\top A\right)^{-1}\left(\frac{1}{m} A^\top \bar b\right) $. For $\frac{1}{m} A^\top \bar b$ we have 
\begin{align*}
\frac{1}{m}A^\top \bar b = \frac{1}{m}A^\top A \left[ 
\begin{array}{cc}
x^o  \\
\|x^o\|^2+\sigma^2
\end{array}
\right]+\frac{1}{\sqrt{m}} O_p\left(1 \right),
\end{align*}
where we have used Lemma~\ref{property_of_bounded_variance}. Thus, we obtain
\begin{equation*}
\sqrt{m}\left( \hat y^{\rm NE0}_m-\left[ 
\begin{array}{cc}
x^o  \\
\|x^o\|^2+\sigma^2
\end{array}
\right]\right) =O_p\left(1 \right), 
\end{equation*}
which implies the $\sqrt{m}$-consistency of $\hat y^{\rm NE0}_m$.

Now we are on the point to show the $\sqrt{m}$-consistency of $\hat y^{\rm NE}_m$. Based on~\eqref{limit_probability}, for any $\epsilon_1>0$, there exists a $M_1>0$ such that for all $m>M_1$, 
\begin{equation*}
P\left( \hat y^{\rm NE}_m=\hat y^{\rm NE0}_m\right) \geq 1-\epsilon_1. 
\end{equation*}
Based on the $\sqrt{m}$-consistency of $\hat y^{\rm NE0}_m$, for any $\epsilon_2>0$, there exists $M_2>0$ and $N>0$ such that for all $m>M_2$,
\begin{equation*}
P\left( \sqrt{m}\left| \hat y^{\rm NE0}_m-\bar y^o\right| \geq N\right)\leq \epsilon_2. 
\end{equation*}
Therefore, for all $m>\max \{M_1,M_2\}$, we have 
{\small
\begin{align*}
&P\left( \sqrt{m}\left| \hat y^{\rm NE}_m-\bar y^o\right| \geq N\right) \\
&= P\left( \hat y^{\rm NE}_m=\hat y^{\rm NE0}_m\right) P\left( \sqrt{m}\left| \hat y^{\rm NE}_m-\bar y^o\right| \geq N \mid \hat y^{\rm NE}_m=\hat y^{\rm NE0}_m\right)  \\
& ~~~+P\left( \hat y^{\rm NE}_m \neq \hat y^{\rm NE0}_m\right) P\left( \sqrt{m}\left| \hat y^{\rm NE}_m-\bar y^o\right| \geq N \mid \hat y^{\rm NE}_m \neq \hat y^{\rm NE0}_m\right) \\
&\leq  \epsilon_2+\epsilon_1.
\end{align*}}
\hspace{-2mm}Since $\epsilon_1$ and $\epsilon_2$ can be arbitrary, we have $\hat y^{\rm NE}_m-\bar y^o=O_p(1/\sqrt{m})$ which completes the proof.

%
%
%
%
%

\ifCLASSOPTIONcaptionsoff
  \newpage
\fi

\small
\bibliographystyle{IEEEtran}
\bibliography{sj_reference}
\end{document}